\documentclass[12pt,a4paper]{article}
\usepackage[T1]{fontenc}
\usepackage[utf8]{inputenc}
\usepackage{lmodern,microtype}
\usepackage[margin=25mm]{geometry}
\usepackage{amsmath,amssymb,amsthm,booktabs,tabularx,array,enumitem,setspace,xcolor}
\usepackage[round,authoryear]{natbib}
\usepackage{caption}
\usepackage[colorlinks=true,linkcolor=blue!40!black,citecolor=blue!40!black,urlcolor=blue!40!black]{hyperref}
\hypersetup{pdftitle={Hidden Eligibility and the Credibility of Adoption},pdfauthor={Georgy Lukyanov and Mikhail Zolotarev}}
\newtheorem{theorem}{Theorem}
\newtheorem{proposition}[theorem]{Proposition}
\newtheorem{lemma}[theorem]{Lemma}

\theoremstyle{definition}\newtheorem{definition}[theorem]{Definition}
\newcommand{\E}{\mathbb E}
\newcommand{\Prb}{\mathbb P}
\newcommand{\ind}{\mathbf 1}
\newcommand{\rb}{\bar r}

\newcommand{\RR}{\mathcal R}
\newcommand{\LL}{\mathcal L}
\newcommand{\cav}{\operatorname{cav}}

\newcolumntype{Y}{>{\raggedright\arraybackslash}X}
\newcolumntype{P}[1]{>{\raggedright\arraybackslash}p{#1}}
\setlist[enumerate]{leftmargin=*,itemsep=3pt,topsep=3pt}
\setlist[itemize]{leftmargin=*,itemsep=3pt,topsep=3pt}
\title{Hidden Eligibility and the Credibility of Adoption}
\author{Georgy Lukyanov\thanks{Toulouse School of Economics, Universit\'e Toulouse Capitole, Toulouse, France. Corresponding author: \href{mailto:georgy.lukyanov@tse-fr.eu}{georgy.lukyanov@tse-fr.eu}. ORCID: \href{https://orcid.org/0009-0005-1672-610X}{0009-0005-1672-610X}.}
\and Mikhail Zolotarev\thanks{International College of Economics and Finance, HSE University, Moscow, Russia.}}
\date{November 2026}
\begin{document}
\maketitle
\begin{abstract}
Subsidies encourage adoption, but they also change what others can learn from it. We study a government that chooses the amount and coverage of an early-adoption rebate and the information released about eligibility. Both early and late adoption create positive externalities. When eligibility is hidden, expanding the programme makes adoption less informative. As the benefit of early adoption rises, the optimal programme jumps from a small one that preserves imitation to a larger one that gives it up. We show that partial disclosure can support intermediate coverage and strictly improve on both optimally chosen opacity and full disclosure. An optimal report certifies the eligibility of some users, leaving the remaining adoptions sufficiently persuasive to induce imitation. We characterise the conditions for this policy and compare coverage under the alternative reporting rules. Analytical extensions allow delivery failures, private information, and heterogeneous adoption costs. The discontinuous change in optimal coverage can persist even when cost heterogeneity makes welfare continuous in coverage.

\end{abstract}

\section{Introduction}\label{sec:introduction}

A household considering a new technology may learn from the decision of an earlier user. An installation suggests that the earlier user received favourable information about the technology. If that user was eligible for a subsidy, however, the same decision may reflect financial support. Expanding a subsidy programme therefore has two effects: it encourages early adoption, but makes each observed adoption less informative to those who decide later.

In this paper, we ask how the government should choose the coverage of such a programme and disclose information about eligibility. Full disclosure helps later users interpret the choices they observe. Yet these users do not take into account the benefits that their own adoption creates for others. The government may consequently prefer a report that preserves some imitation. The question is how this reporting choice affects the scale of the subsidy programme.

We consider two users who adopt in sequence. The early user receives a private signal about the common return to adoption. The late user observes the early decision and a public report about eligibility. Before either user acts, an uninformed government chooses the probability of early eligibility, the rebate amount, and the reporting rule. Both adoptions create a positive external benefit, with an additional benefit from adopting early. Payments are available only to the early cohort, and the report can depend on eligibility alone.

We first characterise the optimal rebate and reporting rule. In the benchmark, a rebate either leaves behaviour unchanged or induces an eligible user to adopt after either private signal. The government can therefore use the smallest amount that induces the latter behaviour. An optimal report certifies eligibility for some users and leaves the others uncertified. All ineligible users remain in the uncertified group, together with some eligible users. When the external benefit of late adoption is sufficiently large, the government leaves as many eligible users in this group as is consistent with the late user's willingness to imitate.

Our main result concerns optimal coverage. Under opacity, adoption ceases to induce imitation once coverage exceeds a threshold. The government then faces a choice between a smaller programme that preserves imitation and a larger one that gives it up. As the additional benefit of early adoption increases, optimal coverage jumps from the first alternative to the second. Partial certification permits an intermediate programme. Certifying some eligible users makes adoption by an uncertified user more persuasive, allowing the government to expand coverage while retaining imitation. We derive a nonempty open set of parameters under which this policy strictly improves on both the best opaque programme and the best programme with full disclosure.

The comparison between opacity and full disclosure also depends on the benefit of early adoption. With a small timing advantage, opacity can support greater coverage because it preserves socially valuable imitation. With a larger advantage, the government may wish to expand, but remains at the credibility threshold to avoid losing late adoption. Full disclosure then permits greater coverage. After the opaque programme expands past the threshold, the comparison can reverse again. We characterise these cases globally, including the two optimal opaque coverages at the point of the jump.

We then examine whether these conclusions depend on discrete choices and indifference. In the benchmark region where partial certification is strictly optimal, slightly increasing the rebate and certifying a slightly larger share of eligible users preserves the strict welfare gain while making all reached decisions strict. We also allow delivery failures, private information of the late user, and heterogeneous late adoption costs. These extensions admit analytical solutions and robustness conditions. In particular, the private signal can affect the late user's action under optimal reporting, and a continuous cost distribution need not remove the discontinuous change in optimal opaque coverage.

A temporary clean-technology pilot illustrates the environment. A predetermined early cohort can receive support, while later households observe installations but fall outside the programme's payment mandate. Evidence of peer effects in solar diffusion motivates the observation of early installations (Bollinger \& Gillingham, \citeyear{BG}). It does not establish that these effects arise solely from Bayesian learning, or that existing programmes practise selective eligibility certification. We use the pilot as a possible institutional design and discuss its restrictions in Section \ref{sec:application}.

The restriction to early payments matters for the result. If the government can correct the late user's adoption incentive directly, the reason to limit disclosure may disappear. We demonstrate this in a simple Pigouvian benchmark. Our analysis concerns the joint choice of coverage and reporting within the early programme.

\subsection{Related literature}\label{sec:literature}

Our paper relates to the literature on subsidies, externalities, and information design. \citet{LX} show how corrective subsidies and taxes affect a sender's information provision, so that the optimal intervention need not equal the externality. In our setting, the payment goes to the early user, while the later observer receives no payment. The government controls both coverage and the disclosure of its eligibility records. We study how these two choices interact: coverage changes the information contained in adoption, and reporting affects how far the programme can expand without losing imitation.

\citet{CFP} show that suitable investment subsidies can implement constrained-efficient acquisition and use of private information in an environment with spillovers and endogenous private information. Their result points to the importance of the available policy instruments. In our model, the government cannot subsidise the late user or disclose the underlying state. These restrictions create a role for eligibility reporting. More broadly, \citet{LSS} study disclosure of hard information without ex-ante commitment and show that externalities generate a conflict between a public communicator and receivers. Here the government commits to an eligibility report before information arrives and chooses programme coverage jointly with that report.

The analysis of reporting uses the posterior-splitting method of \citet{KG}. Here the government's information consists of an eligibility record, and the late user separately observes the early action. These features determine the feasible posterior interval and make its mean depend on coverage. The concavification step follows the standard argument; we then optimise the resulting welfare over coverage. \citet{BC} study competing schools that choose education investment and grading policies; imperfect disclosure can increase investment. Our setting has a single administrator choosing programme coverage. We focus on the discontinuous expansion under opacity and the possibility of intermediate coverage under partial certification.

The paper also relates to observational learning. \citet{BHJ} study long-run efficiency with disclosure and action-contingent transfers, distinguishing budget commitments from actual expenditure. Our model has two users and evaluates welfare under a restriction to early payments. \citet{PRSS} consider hidden revealers who follow their private signals and disrupt cascades. In our setting, a sufficiently large rebate has the opposite informational effect: it makes the eligible user's action independent of her signal. When the late user also has private information, we derive the conditions under which she follows that information or overrides it. A comparison of signal accuracies alone does not determine her decision.

Finally, \citet{LL} study dynamic technology subsidies with forward-looking adoption decisions. In our model, cohort membership and adoption opportunities are fixed. This isolates the informational effect of eligibility, leaving strategic waiting and the choice of a subsidy path outside the analysis.

The rest of the paper is organised as follows. Section \ref{sec:model} presents the model. Section \ref{sec:reporting} derives welfare and the optimal reporting rule. Section \ref{sec:policy} characterises coverage and compares the policies. Section \ref{sec:robustness} studies implementation and robustness, with the longer derivations in the appendices. Section \ref{sec:application} discusses the application and the role of the policy restrictions. Section 7 concludes.

\section{Environment and behaviour}\label{sec:model}

\subsection{Users, information, and public policy}

We consider two users, one early and one late, who each decide whether to adopt. Their choices are $a_E,a_L\in\{0,1\}$. Adoption yields material payoff $\theta-c$, where the common state $\theta\in\{0,1\}$ has prior probability $1/2$ of being one. Non-adoption yields zero. The early user observes a private signal $S\in\{h,\ell\}$ with accuracy
\begin{equation}
 \Prb(S=h\mid\theta=1)=\Prb(S=\ell\mid\theta=0)=q,
 \qquad \frac12<c<q<1.\label{eq:signals}
\end{equation}
Her posterior is $q$ after the high signal and $1-q$ after the low signal. Thus, without a subsidy, she adopts only after the high signal. The late user has no private information in the benchmark; Section \ref{sec:robustness} relaxes this assumption.

Early adoption creates an external benefit $\eta+\tau$, while late adoption creates $\eta$, where $\eta>0$ and $\tau\geq0$. The parameter $\tau$ captures the additional benefit of adopting sooner. These benefits accrue to people outside the pair, so neither user internalises the effect of her decision on others. In particular, the early user's payoff does not depend on the late user's response. The external benefits are known; uncertainty about $\theta$ concerns private returns. Each user has one opportunity to adopt, and the early user cannot wait to become the late user.

Before signals and eligibility are drawn, an uninformed government publicly commits to three objects: coverage $r\in[0,1]$, rebate $s\geq0$, and a reporting rule $\pi(m\mid R)$ on a finite message set. The eligibility indicator $R$ equals one with probability $r$, independently of the state and signal. The early user observes $(S,R)$ and receives the rebate if $R=a_E=1$. After she chooses, a report $m$ is released. Conditional on $R$, the report is independent of the state, signal, and action. The late user observes the policy, the early action, and the report, and then chooses $a_L$.

The report can be drawn when eligibility is assigned and held until the early choice has occurred. The early user does not observe its realisation before acting. Under full disclosure the report reveals $R$; under opacity it is constant. A more general rule can certify eligibility for some users while leaving others unidentified. The government neither observes nor reports the state, and the early action is always observable.\footnote{Eligibility is independent of the state before adoption. A report about eligibility becomes informative about the state when combined with the observed action.} Product performance is not observed before the late decision.

Payments are restricted to the early cohort. Reporting is costless and the government can commit to its rule. We abstract from a seller's price response and from reputational motives of the early user. Consequently, her adoption incentive is independent of the reporting rule: she does not choose her action to influence the late user.

\subsection{Payoffs and equilibrium}

Let $g(r)$ denote real administrative costs, with $g(0)=0$, $g'\geq0$, and $g$ continuously differentiable and strictly convex on $[0,1]$. Let $\lambda\geq0$ be the excess burden of raising public funds and define $G(r)=(1+\lambda)g(r)$. The users' decision payoffs are
\begin{align}
 U_E(a_E\mid S,R)&=a_E\{\Prb(\theta=1\mid S)-c+sR\},\label{eq:ue}\\
 U_L(a_L\mid a_E,m)&=a_L\{p(a_E,m)-c\},\label{eq:ul}
\end{align}
where $p(a_E,m)$ is the posterior after the information actually observed. The government maximises
\begin{equation}
 W(r,s,\pi)=\E[a_E(\theta-c+\eta+\tau)+a_L(\theta-c+\eta)]
               -G(r)-\lambda\E[sRa_E].\label{eq:welfare}
\end{equation}
Taxpayers outside the pair finance the programme. A rebate is a transfer from taxpayers to the early user, so its receipt and payment cancel in social welfare. The government does, however, incur real administrative costs and the excess burden of financing both administration and rebates. These are the two cost terms in \eqref{eq:welfare}.

\begin{definition}
For\label{def:equilibrium} a committed policy, a selected continuation equilibrium consists of an early strategy maximising \eqref{eq:ue}, late beliefs satisfying Bayes' rule at positive-probability histories, and a late strategy maximising \eqref{eq:ul}. Both users adopt at indifference. At impossible histories, beliefs and actions satisfy sequential rationality. An optimal committed policy maximises \eqref{eq:welfare} over the selected continuation equilibria.
\end{definition}

We thus use a weak perfect Bayesian continuation equilibrium, with adoption selected at indifference, preceded by the government's commitment. The policy cannot signal private information of the government, since it has none. Off-path beliefs and message labels need not be unique. The tie selection is consequential because late adoption has a positive external benefit even when its private surplus is zero. We return to this issue in Section \ref{sec:strict} and Appendix \ref{app:mixing}, allowing the competing policies to use other early best responses as well.

\subsection{The rebate amount}

\begin{lemma}\label{lem:rebate}
Let $s_0=q+c-1>0$. Every policy is weakly dominated by either no programme or a policy with $s=s_0$. Under this inducing rebate,
\begin{equation}
 a_E=\ind\{S=h\text{ or }R=1\}.\label{eq:early}
\end{equation}
For $\lambda>0$ and $r>0$, an inducing rebate larger than $s_0$ is strictly dominated by $s_0$ with the same coverage and report.
\end{lemma}

\begin{proof}
A high-signal user adopts without support. A low-signal eligible user adopts exactly when $1-q-c+s\geq0$, or $s\geq s_0$. Below this amount, the action reveals the signal and the independent eligibility report adds no information. Such a programme reproduces the no-programme allocation while adding $G(r)+\lambda sr/2$ to welfare costs. At or above $s_0$, every eligible user adopts; actions and information are independent of the rebate amount. Reducing the rebate to $s_0$ saves $\lambda r(s-s_0)$. At $r=0$ its nominal amount is immaterial.
\end{proof}

Lemma \ref{lem:rebate} uses the fact that all early users face the same cost and receive a binary signal. With heterogeneous early costs, different rebate amounts could induce different groups to adopt. Here the government only needs to decide whether to induce adoption after the low signal. If $\lambda=0$, all inducing amounts give the same welfare. Paying $s_0+\varepsilon_s$ makes inducement strict at an additional welfare cost of at most $\lambda\varepsilon_s$.

\section{Eligibility reporting and welfare}\label{sec:reporting}

\subsection{The information in an observed adoption}

Set $s=s_0$. For a report $m$, define the unconditional event masses
\begin{equation}
 x_m=\frac{1-r}{2}\pi(m\mid0),\qquad y_m=r\pi(m\mid1).
 \label{eq:masses}
\end{equation}
Table \ref{tab:histories} gives the probabilities and beliefs at each possible history.
\begin{center}\begin{minipage}{\textwidth}\captionof{table}{Probabilities and beliefs at public histories}\label{tab:histories}

\begin{tabular}{@{}llll@{}}
\toprule
Event & Probability & Early action & Posterior on $\theta=1$\\
\midrule
$R=0,S=h,m$ & $x_m$ & 1 & $q$\\
$R=0,S=\ell,m$ & $x_m$ & 0 & $1-q$\\
$R=1,m$ & $y_m$ & 1 & $1/2$\\
\bottomrule
\end{tabular}
\end{minipage}\end{center}
In particular, $x_m+y_m$ is the probability of $(a_E=1,m)$, whereas $2x_m+y_m$ is the probability of $m$. After adoption,
\begin{equation}
 p(1,m)=\frac{qx_m+y_m/2}{x_m+y_m},\qquad
 \sum_mx_m=\frac{1-r}{2},\quad \sum_my_m=r.\label{eq:posterior}
\end{equation}
After non-adoption the posterior is $1-q<c$, so the late user declines regardless of the report. Zero-probability histories require no Bayesian fraction.

At a fixed coverage, full disclosure is more informative about the state than any other eligibility report. To obtain $(a_E,m)$ from $(a_E,R)$, one need only randomise the report conditional on $R$, independently of the state. This is Blackwell dominance: the less informative observation is a garbling of the more informative one. The late user's expected private payoff weakly increases with information, because $\max\{p-c,0\}$ is convex in her posterior $p$. Her contribution to social welfare is $(p-c+\eta)\ind\{p\geq c\}$, which need not be convex. The externality therefore makes the social value of disclosure different from its private value.

\subsection{Welfare}

Before administration and financing costs, expected early welfare is
\begin{equation}
 E_E(r)=\frac{q-c+\eta+\tau}{2}
             +\frac r2(1-q-c+\eta+\tau).\label{eq:early-welfare}
\end{equation}
If report $m$ induces late adoption after $a_E=1$, its contribution is
\begin{equation}
 (q-c+\eta)x_m+(1/2-c+\eta)y_m.\label{eq:late-social}
\end{equation}
Both expressions use the joint probabilities in Table \ref{tab:histories}, so they already account for the dependence between the two decisions. We can add the users' contributions even when their actions are correlated. The government's objective is their social contribution in \eqref{eq:welfare}.

For the same reason, $\E[a_E+a_L]=\E[a_E]+\E[a_L]$. This counts two adoptions when both users adopt. The probability that at least one user adopts is a different object.

Under opacity, the adoption posterior and its credibility threshold are
\begin{equation}
 \mu(r)=\frac{q+r(1-q)}{1+r},\qquad
 \bar r=\frac{q-c}{q+c-1},\qquad
 \mu(r)\geq c\ \Longleftrightarrow\ r\leq\bar r.\label{eq:mu}
\end{equation}
Under full disclosure, the late user adopts only when the early user adopted without being eligible. Let $W_T$ and $W_O$ denote welfare under full disclosure and opacity, respectively. The preceding calculations give
\begin{align}
 W_T(r)&=E_E(r)+\frac{1-r}{2}(q-c+\eta)-G(r)-\lambda s_0r,\label{eq:rawT}\\
 W_O(r)&=E_E(r)+\ind\{r\leq\bar r\}
       \frac{q-c+\eta+r(1-q-c+\eta)}2-G(r)-\lambda s_0r.
       \label{eq:rawO}
\end{align}
At zero coverage both users adopt after the high early signal, an event of probability $1/2$. Their expected welfare is consequently
\begin{equation}
 W_0=q-c+\eta+\frac\tau2.\label{eq:W0}
\end{equation}

\subsection{The optimal report}

An eligible user adopts even after a low signal, so her adoption by itself leaves the late user at the prior. Imitating this user entails expected private loss $d=c-1/2>0$. The government nevertheless gains $\eta-d$ from such imitation. It may therefore wish to leave some eligible users unidentified, provided the uncertified group also contains enough informative adoptions. The late user's expected private gain after report $m$, multiplied by the probability of that observation, is $(q-c)x_m-dy_m$. This expression gives the constraint on how much eligible adoption can remain in the group.

\begin{proposition}\label{prop:closed}
At every coverage an optimal report can use two messages, $u$ (uncertified) and $e$ (eligible), with
\begin{equation}
 \pi(u\mid0)=1,\qquad \pi(u\mid1)=\sigma,\qquad
 \pi(e\mid1)=1-\sigma.\label{eq:certificate}
\end{equation}
Let $y=r\sigma$. The government's complete problem is equivalent in value to
\begin{align}
 \max_{r,y}\quad &W_0+br-G(r)+(\eta-d)y,\label{eq:program}\\
 \text{subject to}\quad &0\leq r\leq1,\quad 0\leq y\leq r,\quad
                   2dy\leq(q-c)(1-r),\label{eq:obedience}\\
 b&=\frac\tau2-(q-\tfrac12)-\lambda s_0.\label{eq:notation}
\end{align}
The late user adopts only after uncertified adoption. An optimum exists, and its coverage is unique.
\end{proposition}

\begin{proof}
Merge all messages inducing late adoption into one and all other messages into another. Weighted averages preserve the respective inequalities $p\geq c$ and $p<c$, and decisions after non-adoption remain unchanged. Move every remaining ineligible report into the adoption-inducing message. This raises its posterior and adds $q-c+\eta>0$ per unit of moved adoption mass. The other message then certifies eligibility and has posterior $1/2<c$. If no original message induces adoption, create the first message from all ineligible histories.

The resulting uncertified observation has ineligible adoption mass $(1-r)/2$ and eligible mass $y$. Its incentive constraint is exactly \eqref{eq:obedience}. Conversely, each feasible pair is implemented by \eqref{eq:certificate}, with empty-group probabilities immaterial at $r=0,1$. Adding \eqref{eq:early-welfare} and \eqref{eq:late-social} and subtracting costs gives \eqref{eq:program}. Lemma \ref{lem:rebate} covers all rebate amounts. The feasible set is compact and convex. The objective is continuous and concave, with strict concavity in coverage, proving existence and uniqueness of coverage.
\end{proof}

The certificate identifies eligibility. It does not establish whether the rebate caused an individual to adopt.\footnote{An eligible user with a high signal would have adopted without support. The certificate identifies her incentive regime, while the individual counterfactual remains unobserved.} An uncertified user may be either eligible or ineligible. In both cases the report is drawn independently of the action, as required by the model.

To describe the solution, define
\begin{equation}
 \chi=\frac{q-c}{2d},\qquad \zeta=\eta-d,\qquad
 \bar y(r)=\min\{r,\chi(1-r)\}.\label{eq:capacity}
\end{equation}
The two bounds on $y$ meet at $\bar r$. If $\eta<d$, optimal reporting sets $y=0$ and fully discloses eligibility. If $\eta>d$, it sets $y=\bar y(r)$: opacity is optimal below $\bar r$, while a strictly partial certificate is optimal above it and below one. At $\eta=d$, all feasible retained masses have the same value. Optimised reporting therefore gives
\begin{equation}
 V(r)=W_0+br-G(r)+(\zeta)_+\bar y(r).\label{eq:V}
\end{equation}
The reporting choice has therefore reduced the government's problem to maximising a continuous, strictly concave function of coverage. The next section solves this remaining choice and compares it with the two restricted reporting rules.

\section{Programme coverage and the value of certification}\label{sec:policy}

\subsection{Choosing coverage and reporting together}

Once the reporting rule is fixed within a given region, welfare has a linear benefit term and a convex cost. We will repeatedly solve this simple problem. For a benefit slope $v$ and an allowed coverage interval $[l,u]$, write
\begin{equation}
 \RR(v;l,u)=\arg\max_{r\in[l,u]}\{vr-G(r)\}
 =\begin{cases}
 l,&v\leq G'(l),\\
 (G')^{-1}(v),&G'(l)<v<G'(u),\\
 u,&v\geq G'(u).
 \end{cases}\label{eq:R}
\end{equation}
An interior solution equates the marginal benefit $v$ to $G'(r)$; otherwise the optimum is at an endpoint. Strict convexity makes $G'$ strictly increasing, so this solution is unique. The inverse in \eqref{eq:R} is used only on its range. Under full disclosure,
\begin{equation}
 r_T=\RR(b;0,1).\label{eq:rT}
\end{equation}

\begin{proposition}\label{prop:joint}
If $\eta\leq d$, full disclosure is optimal and $r^*=r_T$. If $\eta>d$, define
\begin{equation}
 b_0=G'(0)-\zeta,\quad b_L=G'(\bar r)-\zeta,\quad
 b_C=G'(\bar r)+\chi\zeta,\quad b_F=G'(1)+\chi\zeta.
 \label{eq:thresholds}
\end{equation}
Table \ref{tab:joint} gives the unique joint-optimal coverage and one optimal report:
\begin{center}\begin{minipage}{\textwidth}\captionof{table}{Optimal coverage and reporting}\label{tab:joint}

\small
\begin{tabularx}{\textwidth}{@{}P{32mm}P{52mm}Y@{}}
\toprule
Condition & Coverage & Report\\
\midrule
$b\leq b_0$ & $r^*=0$ & No programme\\
$b_0<b<b_L$ & $r^*=(G')^{-1}(b+\zeta)$ & Opacity\\
$b_L\leq b\leq b_C$ & $r^*=\bar r$ & Opacity at the threshold\\
$b_C<b<b_F$ & $r^*=(G')^{-1}(b-\chi\zeta)$ & Partial certification\\
$b\geq b_F$ & $r^*=1$ & Reports immaterial\\
\bottomrule
\end{tabularx}
\end{minipage}\end{center}
\end{proposition}

All proofs omitted from the main text are in the appendices. To understand the plateau at $\bar r$, consider a government that has expanded coverage to the point where the late user is just willing to imitate. Reducing coverage would give up useful early adoption. Expanding it requires certifying some eligible users, whose adoption will then cease to induce imitation. For a range of timing benefits, neither change is worthwhile. Once the benefit of early adoption is large enough, the government expands and certifies the fraction
\begin{equation}
 1-\frac{\chi(1-r^*)}{r^*}\in(0,1),\label{eq:certprob}
\end{equation}
which rises with coverage. At full coverage, every early user adopts and no report can restore information about the state.

\subsection{Coverage under opacity}

Let $H=(q-c+\eta)/2$. Equations \eqref{eq:rawT}--\eqref{eq:rawO} simplify to
\begin{align}
 W_T(r)-W_0&=br-G(r),\label{eq:WT}\\
 W_O(r)-W_0&=\begin{cases}
 (b+\zeta)r-G(r),&0\leq r\leq\bar r,\\
 (b+H)r-G(r)-H,&\bar r<r\leq1.
 \end{cases}\label{eq:WO}
\end{align}
Below $\bar r$, the late user imitates every adoption; above it, she never imitates. This change in behaviour explains the two branches of \eqref{eq:WO}. At the threshold, imitation gives her zero private surplus but still benefits others. Its disappearance therefore causes a downward jump in welfare of
\begin{equation}
 J=\zeta\bar r+H(1-\bar r)=\frac{\eta(1+\bar r)}2>0.
 \label{eq:jump}
\end{equation}

\begin{proposition}\label{prop:opaque}
For every $\eta>0$, set
\begin{align}
 r_l(b)&=\RR(b+\zeta;0,\bar r),& L(b)&=(b+\zeta)r_l-G(r_l),\notag\\
 r_h(b)&=\RR(b+H;\bar r,1),& U(b)&=(b+H)r_h-G(r_h)-H.
 \label{eq:branches}
\end{align}
The optimised opaque value is $W_0+\max\{L(b),U(b)\}$. There is a unique switch $b^\dagger>G'(0)$ such that its optimal coverage correspondence is
\begin{equation}
 \mathcal O(b)=\begin{cases}
 \{r_l(b)\},&b<b^\dagger,\\
 \{r_l(b),r_h(b)\},&b=b^\dagger,\\
 \{r_h(b)\},&b>b^\dagger.
 \end{cases}\label{eq:O}
\end{equation}
At the switch $r_h>\bar r\geq r_l$. A winning or tied high candidate always lies strictly above $\bar r$, so closing its candidate interval creates no false optimum. If $\eta\leq d$, then $b^\dagger<G'(1)$; if $\eta=d$, also $b^\dagger>G'(\bar r)$.
\end{proposition}

Proposition \ref{prop:opaque} compares the best programme on each side of the threshold. As the timing benefit rises, the larger programme becomes more attractive, and the government eventually switches to it. At the switch both programmes are optimal; an intermediate coverage is not. The switch is economically feasible: varying $\tau\geq0$ gives $b\geq b_{\min}=-(q-1/2)-\lambda s_0<0$, while $b^\dagger>G'(0)$. The formula in \eqref{eq:branches} characterises the choice even when $G'$ has no elementary inverse, since it determines both candidates and compares their values.

\subsection{Certification permits intermediate coverage}

\begin{theorem}\label{prop:main}
Partial certification strictly outperforms both optimised full disclosure and optimised opacity if and only if
\begin{equation}
 \eta>d,\qquad G'(\bar r)<b-\chi(\eta-d)<G'(1).
 \label{eq:strictregion}
\end{equation}
For fixed admissible $(q,c,\eta,\lambda,G)$ with $\eta>d$, this is the nonempty open interval
\begin{align}
 2\{q-\tfrac12+\lambda s_0+G'(\bar r)+\chi(\eta-d)\}
 &<\tau\notag\\
 &<2\{q-\tfrac12+\lambda s_0+G'(1)+\chi(\eta-d)\}.
 \label{eq:tauinterval}
\end{align}
Moreover,
\begin{equation}
 b_C<b^\dagger<b_F.\label{eq:switchbounds}
\end{equation}
At the opaque switch, the low optimum is $\bar r$, and the joint optimum lies strictly between the two opaque optima. Away from the switch within this interval,
\begin{align}
 b_C<b<b^\dagger &: \quad r_O=\bar r<r^*<r_T\leq1,\label{eq:beforejump}\\
 b^\dagger<b<b_F &: \quad \bar r<r^*<r_T\leq r_O.\label{eq:afterjump}
\end{align}
In the second line, equality $r_T=r_O$ is possible only at one.
\end{theorem}

The intuition for Theorem \ref{prop:main} is that partial certification allows the government to expand without giving up imitation altogether. Under opacity, it remains at the credibility threshold until the extra benefit of early adoption justifies losing late adoption. It then expands discontinuously. Under certification, the government expands gradually, increasing the share of eligible users identified as coverage rises. The theorem compares the best programmes under each rule, allowing coverage to adjust in all three cases.

The two conditions in \eqref{eq:strictregion} have distinct roles. The first makes imitation of eligible users socially valuable. The second ensures that the benefit of early adoption justifies coverage above the credibility threshold, but below one. In particular, a timing advantage is necessary: at $\tau=0$, $b<0$, whereas the second condition requires $b>G'(\bar r)+\chi(\eta-d)\geq0$. A large common external benefit can justify an opaque programme at or below $\bar r$, but cannot on its own make partial certification strictly preferable to both alternatives.

\subsection{Comparing opacity with full disclosure}

Theorem \ref{prop:main} compares coverages where partial certification is strictly optimal. We now give the comparison between opacity and full disclosure throughout the parameter domain. Two forces matter. On the low branch, opacity preserves imitation of eligible users. On the high branch, imitation is already lost, so expanding coverage no longer reduces late adoption. The optimal opaque programme is whichever branch gives higher welfare.

\begin{proposition}\label{prop:order}
Every winning or tied high opaque candidate satisfies $r_h\geq r_T$, strictly if $r_T<1$. For the low candidate: if $\zeta<0$, then $r_l\leq r_T$, strictly if $r_T>0$; if $\zeta=0$, then $r_l=\min\{r_T,\bar r\}$; if $\zeta>0$, then $r_l\geq r_T$ when $r_T\leq\bar r$ and $r_l=\bar r<r_T$ otherwise. The global candidate is selected by \eqref{eq:O}.

For $\zeta>0$, Table \ref{tab:order} gives the complete global comparison.
\begin{center}\begin{minipage}{\textwidth}\captionof{table}{Opaque coverage relative to full disclosure}\label{tab:order}

\small
\begin{tabularx}{\textwidth}{@{}P{48mm}Y@{}}
\toprule
Condition & Optimal opaque coverage relative to full disclosure\\
\midrule
$b<G'(\bar r)$ & $r_O\geq r_T$; equality only when both are zero\\
$b=G'(\bar r)$ & $r_O=r_T=\bar r$\\
$G'(\bar r)<b<b^\dagger$ & $r_O=\bar r<r_T$\\
$b=b^\dagger$ & $\mathcal O=\{\bar r,r_h\}$, with $\bar r<r_T\leq r_h$\\
$b>b^\dagger$ & $r_O\geq r_T$; equality only when both are one\\
\bottomrule
\end{tabularx}
\end{minipage}\end{center}
For every fixed $\eta>0$ and other primitives except $\tau$, nonempty open intervals of timing advantages yield each strict ordering, $r_O<r_T$ and $r_O>r_T$.
\end{proposition}

The comparison is stated under the benchmark selection, in which every eligible user adopts at $s_0$. Appendix \ref{app:mixing} allows a low-signal eligible user to mix at that rebate and gives the adjusted coverage comparison. Mixing can improve the full-disclosure benchmark outside the region in Theorem \ref{prop:main}. Within that region, it changes neither the strict welfare ranking nor the coverages in the theorem.

\subsection{Comparative statics}\label{sec:statics}

An increase in the timing advantage $\tau$ makes early adoption more valuable and weakly raises optimal coverage. An increase in the common benefit $\eta$, holding $\tau$ fixed, has a different effect. For $\eta\leq d$, the optimal report fully reveals eligibility and coverage remains $r_T$. For $\eta>d$, greater weight on the common benefit moves coverage towards $\bar r$: it rises if $r_T<\bar r$ and falls if $r_T>\bar r$. This follows from the two benefit slopes $b+\zeta$ and $b-\chi\zeta$ in \eqref{eq:V}. Coverage reaches the threshold when
\begin{equation}
 \eta-d\geq\max\left\{0,G'(\bar r)-b,\frac{b-G'(\bar r)}{\chi}\right\}.
 \label{eq:etakink}
\end{equation}
The intuition comes from total adoption. Under optimal reporting with $\eta>d$, its expected value is $1+\bar y(r)$, which is largest at $\bar r$. Raising the common external benefit increases the value of this total, whereas raising the timing advantage shifts the government's preference towards early adoption.

Higher costs reduce coverage. Holding the real administrative technology $g$ fixed, an increase in the financing burden $\lambda$ lowers $b$ and raises $G'=(1+\lambda)g'$. Both changes weakly reduce optimal coverage. A pointwise increase in marginal administrative costs has the same effect. These comparative statics include the plateau and the corner solutions. Section \ref{sec:application} considers changes in environmental benefits that raise $\eta$ and $\tau$ together.

\section{Implementation and extensions}\label{sec:robustness}

\subsection{Strict decisions and early indifference}\label{sec:strict}

The benchmark selects adoption when a user is indifferent. This guarantees that the government's optimum is attained. If the late user instead declined at equality, an opaque programme could approach its best value as coverage rises towards $\bar r$, yet fail to attain it. The difficulty arises because imitation still creates an external benefit when its private surplus is zero.

The welfare gain can nevertheless be obtained with strict decisions. At a point in \eqref{eq:strictregion}, let $\Gamma>0$ be the gap between joint welfare and the better restricted optimum. Keep $r^*$ fixed, pay $s_0+\varepsilon_s$, and reduce the uncertified eligible mass from $y^*=\chi(1-r^*)$ to $y^*-\delta_y$, with $0<\delta_y<y^*$. Early inducement is strict. The posterior after uncertified adoption exceeds $c$, while all other reached late decisions are strict non-adoptions. The welfare loss is exactly
\begin{equation}
 \lambda r^*\varepsilon_s+\zeta\delta_y.\label{eq:strictloss}
\end{equation}
Choosing it below $\Gamma$ preserves the strict advantage.

For a fair comparison, opacity and full disclosure must also be allowed to use the early user's other best responses at $s_0$. Appendix \ref{app:mixing} shows that mixing cannot improve the joint optimum or the best opaque programme. It can improve full disclosure elsewhere, but not in \eqref{eq:strictregion}. The strict gain just constructed therefore survives comparison with these alternatives. Lemma \ref{lem:mixing} extends this argument to the information and cost specifications below.

\subsection{Private late information and delivery failures}\label{sec:general}

We now allow delivery failures, private late information, and heterogeneous late costs. A rebate may fail to reach an eligible user. Let $D$ be an indicator that equals one with probability $\rho\in(0,1]$, independently of coverage, eligibility, the state, signals, costs, and reporting randomisation. The early user observes $D$. Only a usable rebate is paid, and raising its nominal amount cannot overcome a delivery failure. Her payoff becomes $a_E\{\Prb(\theta=1\mid S)-c+sRD\}$. At the inducing rebate she adopts if $S=h$ or $RD=1$. The public report still concerns eligibility $R$, rather than actual delivery.

The late user may also have her own information and a privately known adoption cost. She observes a finite signal $Z$ with state-conditional probabilities $f_\theta(z)$ and a cost $C$ with distribution $F_C$ on a compact subset of $(0,1)$. Conditional on the state, $Z$ is independent of the early signal and all programme randomisation. Cost is independent of all these variables and the state. Early cost remains $c$. At public posterior $p$, define
\begin{equation}
 w_z(p)=pf_1(z)+(1-p)f_0(z),\qquad
 p_z(p)=\frac{pf_1(z)}{w_z(p)}.\label{eq:privateposterior}
\end{equation}
Zero-weight terms are omitted. The actual payoff is $a_L\{p_z(p)-C\}$, so the decision is $a_L=\ind\{p_z(p)\geq C\}$. Welfare replaces the late cost in \eqref{eq:welfare} by $C$ and the payment by $sRDa_E$. Before observing $Z,C$, the late social contribution is
\begin{equation}
 \LL(p)=\sum_z w_z(p)\int_{t\leq p_z(p)}(p_z(p)-t+\eta)\,dF_C(t).
 \label{eq:socialvalue}
\end{equation}
Expression \eqref{eq:socialvalue} averages the social contribution before the late user's private information and cost are observed. Her own choice is made after observing both, using $p_z(p)$ and $C$. The formula requires no ordering of signal labels. A monotone likelihood-ratio condition would only be needed to express her decision as a cutoff in such an ordering.

For a fully inducing early policy, put
\begin{equation}
 p_\rho=\frac{q+\rho(1-q)}{1+\rho},\quad a_\rho=\frac{1+\rho}{2},\quad
 w_\rho(r)=\frac{1+\rho r}{2}.\label{eq:delivery-posteriors}
\end{equation}
An eligible adoption now has posterior $p_\rho$: it is informative when delivery sometimes fails. An ineligible adoption still has posterior $q$, and any non-adoption has posterior $1-q$. Thus eligibility reporting can generate posteriors between $p_\rho$ and $q$. Let $\overline{\LL}_\rho=\cav_{[p_\rho,q]}\LL$ be the smallest concave majorant on the indicated interval.

\begin{proposition}\label{prop:general}
For $0<r<1$, a finite distribution of public posteriors after adoption is feasible if and only if its support is in $[p_\rho,q]$ and its mean is $\mu(\rho r)$. At $r=0,1$ it is degenerate at $q,p_\rho$, respectively. Optimised reporting gives
\begin{align}
 \mathcal V_\rho(r)&=E_0+\rho r E_1-G(r)-\lambda s_0\rho r\notag\\
 &\quad+w_\rho(r)\overline{\LL}_\rho(\mu(\rho r))
       +\frac{1-\rho r}{2}\LL(1-q),\label{eq:general-program}\\
 E_0&=\frac{q-c+\eta+\tau}{2},\qquad
 E_1=\frac{\eta+\tau-s_0}{2}.\notag
\end{align}
Two messages suffice. The joint optimum exists and has unique coverage, which is continuous and weakly increasing in $\tau$.
\end{proposition}

Proposition \ref{prop:general} separates the two choices. At each coverage, the government chooses a distribution of posteriors with the stated mean; every feasible distribution can be implemented by a rule that depends on eligibility alone. It then chooses coverage using \eqref{eq:general-program}. The reporting value is concave in coverage, and strictly convex administrative costs make the optimal coverage unique. The two messages need not take the form of an eligibility certificate when the late user has informative private information. Non-adoption remains observable throughout. If $\rho=0$, no action is changed by the programme and zero coverage is optimal.

\subsection{Three analytical cases}

Appendix \ref{app:exact} solves three cases explicitly. Consider first delivery failures alone, with no late private signal and cost $c$. If $\rho\leq\bar r$, then $p_\rho\geq c$: even a known eligible adoption is sufficiently informative to induce imitation. Reporting is then irrelevant. For $\rho>\bar r$, Proposition \ref{prop:rho-exact} gives the coverage regimes and the conditions for partial certification. Delivery failures weaken the reporting trade-off because they leave eligible users' decisions more responsive to private information.

Second, for a finite late signal and deterministic cost, its public-posterior threshold after realisation $z$ is
\begin{equation}
 t_z=\frac{cf_0(z)}{(1-c)f_1(z)+cf_0(z)}.\label{eq:signalcutmain}
\end{equation}
If $p\geq\max_z t_z$, the late user adopts after every private signal; if $p<\min_z t_z$, she never adopts. Between these bounds her signal affects the decision. Its use therefore depends on the public posterior and adoption cost as well as signal accuracy. Proposition \ref{prop:finite} solves the policy problem by comparing finitely many analytical candidates. Appendix \ref{sec:active-signal} gives a nonempty parameter family in which the private signal affects behaviour after an optimal public message, while selective reporting remains strictly better than either restricted rule.

Third, let late costs be uniform on $[c-\varepsilon,c+\varepsilon]$, with $p_\rho<c-\varepsilon<c+\varepsilon<q$. Proposition \ref{prop:uniform} solves for the optimal certificate, allowing delivery failures as well. When partial certification is optimal, uncertified adoption induces posterior $c+\varepsilon$. All cost types then adopt, and only the endpoint type is indifferent. Since this type has probability zero, the welfare gain persists without a positive mass of indifferent late users.

\subsection{Robustness of the welfare and coverage comparisons}\label{sec:local-summary}

We can also vary delivery, private information, and costs together. Fix a benchmark satisfying \eqref{eq:strictregion}. Proposition \ref{prop:robust} shows that the strict welfare advantage persists when delivery failures are sufficiently infrequent, costs are sufficiently close to $c$, and late private signals are sufficiently weak. More precisely, costs lie in $[c-\varepsilon,c+\varepsilon]$ and each private likelihood ratio lies in $[e^{-\nu},e^\nu]$. The cost distribution may be asymmetric. The signal restriction rules out rare but arbitrarily informative observations.

The proof bounds each policy's value by the values in two models with deterministic late costs. These bounds establish convergence of welfare and of optimal joint coverage. They also establish convergence of the restricted optima away from the opaque switch. At the switch, opacity has two optimal coverages, so convergence of its value does not select one of them.

Cost heterogeneity has a further implication. If the cost distribution is continuous, late users stop imitating at different public beliefs, making opaque welfare continuous in coverage. The government's optimal coverage can still jump. Proposition \ref{prop:jump} shows that, near the benchmark switch, it continues to choose between two separated coverage intervals. There is a unique nearby benefit threshold at which both intervals contain an optimum. Optimal coverage under joint reporting remains unique, varies continuously, and lies between them with higher welfare. This is a local result for sufficiently small perturbations; it does not imply a single switch for arbitrary cost distributions.

\section{Application and limits of the policy instruments}\label{sec:application}

\subsection{A temporary technology pilot}

Consider a municipal or regional pilot that subsidises a predetermined early cohort of clean-technology adopters. Later households cannot receive payments from this programme, but they can observe an early installation before its performance becomes known. Coverage is determined by a lottery among prequalified participants. The private signal concerns a common commercial component of returns, such as financing or resale value, while the environmental benefit is known separately.

In this interpretation, the administrator releases information from the eligibility records. A certificate states that a user was eligible for support. It neither certifies the technology's quality nor establishes that the installation would not have occurred without the rebate. Users without a certificate can also be eligible. The administrator must be able to commit to the reporting rule and randomise independently of installation choices. Individual eligibility must not already be publicly observable: a public register revealing every user's status would remove the hidden-eligibility channel, although aggregate expenditure figures need not do so.

We use this pilot to illustrate a possible design. The peer effects documented by \citet{BG} support the relevance of observing installations, but do not establish hidden eligibility or selective certification in an existing programme. Cohort membership and adoption opportunities are also fixed here. This distinguishes the interpretation from a dynamic subsidy programme in which users can wait for future support, as in \citet{LL}.

\subsection{The timing conditions}

Suppose adoption provides external flow benefit $e>0$, and let $D_E>D_L>0$ be discounted benefit durations. Then
\begin{equation}
 \eta=eD_L,\qquad \tau=e(D_E-D_L)=\beta\eta,\qquad
 \beta=\frac{D_E-D_L}{D_L}.\label{eq:duration}
\end{equation}
The extra early benefit reflects the longer period over which it accrues. We retain the model's equal normalised private payoffs across cohorts, so this interpretation does not amount to a complete lifecycle investment model. Substituting \eqref{eq:duration} into the coverage condition gives
\begin{equation}
 b-\chi(\eta-d)=\eta(\beta/2-\chi)-(\lambda+1/2)s_0.
 \label{eq:duration-slope}
\end{equation}
Consequently, strict partial certification requires $\beta>2\chi$. With $K=\beta/2-\chi>0$, the remaining condition is
\begin{equation}
 \max\left\{d,\frac{G'(\bar r)+(\lambda+1/2)s_0}{K}\right\}
 <\eta<\frac{G'(1)+(\lambda+1/2)s_0}{K}.
 \label{eq:application}
\end{equation}
The interval is nonempty exactly when its upper endpoint exceeds $d$. The other lower bound is strictly below that endpoint because $G'(1)>G'(\bar r)$.

Partial certification requires both a sufficiently large timing advantage and $\eta>d$. Within this domain, the bounds involving marginal administrative costs determine whether the government stops at $\bar r$, chooses intermediate coverage, or covers everyone. A rise in the environmental flow benefit increases both $\eta$ and $\tau$. Along $\tau=\beta\eta$, optimal coverage is weakly increasing: the benefit slope below the threshold is $\eta(1+\beta/2)-(1+\lambda)s_0$, while the slope above it is \eqref{eq:duration-slope}. If $\beta\leq2\chi$, the latter is negative and coverage never exceeds $\bar r$ once $\eta>d$. If $\beta>2\chi$, both slopes increase. This comparison differs from raising $\eta$ alone while keeping the early lead fixed.

The separation between private returns and external benefits is a restriction of this interpretation. If uncertain physical output also changes environmental benefits, the constant $\eta$ is only an approximation. State-dependent emissions benefits would require extending the model.

\subsection{Later subsidies and richer records}\label{sec:late-subsidy}

The restriction to early payments is central to the informational conflict. Suppose instead that the government can subsidise late adoption. To see the implication, set $\lambda=0$ and allow a subsidy of $\eta$ for late adoption. The late user then adopts when $p-c+\eta\geq0$, aligning her decision with social welfare. Her social continuation value becomes
\[
 \max\{p-c+\eta,0\},
\]
which is convex in the posterior. Full eligibility disclosure now weakly improves welfare at every fixed early policy. Since early behaviour is independent of the late decision, the government no longer needs to withhold information to encourage socially valuable late adoption. Allowing later transfers with distortionary finance would lead to a different policy problem.

The information available for reporting matters as well. Under delivery failure, suppose the government can report actual usable support $RD$ rather than eligibility alone. With effective coverage $x=\rho r$, the problem reduces to the full-delivery model on $x\in[0,\rho]$, with administrative cost $G(x/\rho)$ and financing burden $\lambda s_0x$. The richer instrument can reproduce any eligibility-only report, so it weakly increases the optimised value. This does not mean that full disclosure of actual receipt is always optimal.

Finally, we assume that the administrator can commit to its reporting rule. It cannot revise the report after seeing adoption, and we do not model the strategic release of performance data. We also restrict the late user to observing one early action. Observing an entire early cohort would change what can be inferred about the state and would require a separate analysis.

\section{Conclusion}

We have studied how disclosure about subsidy eligibility affects the scale of an adoption programme. A larger programme encourages early adoption but makes it less informative to later users. When late adoption has positive external benefits, the government may prefer to preserve some imitation. Under opacity, this creates a choice between a small programme that retains imitation and a larger one that gives it up. Optimal coverage can jump between these alternatives as the benefit of early adoption rises.

Partial eligibility certification permits an intermediate choice. By identifying some eligible users, the government makes adoption by the remaining users sufficiently persuasive to retain imitation. We derive conditions under which this policy strictly improves on both optimally chosen opacity and full disclosure. Analytical extensions show that the result can survive delivery failures, private late information, and heterogeneous costs. Its policy interpretation rests on the available instruments: payments are restricted to early users, reports concern eligibility, and the government can commit. Within these restrictions, coverage and disclosure are best chosen together.

\newpage

\appendix

\section{Proofs for the benchmark policy comparison}\label{app:policy}

\subsection{Proof of Proposition \ref{prop:joint}}
\begin{proof}
For $\zeta\leq0$, \eqref{eq:V} reduces to $W_0+br-G(r)$. For $\zeta>0$, its derivative away from $\rb$ is
\[
 V'(r)=\begin{cases}
 b+\zeta-G'(r),&r<\rb,\\
 b-\chi\zeta-G'(r),&r>\rb.
 \end{cases}
\]
The derivative falls by $(1+\chi)\zeta$ at $\rb$. The maximiser is zero when the right derivative at zero is nonpositive, one when the left derivative at one is nonnegative, and an interior root when it exists within a branch. At $\rb$, the left derivative must be nonnegative and the right derivative nonpositive. These conditions give exactly the table. Strict concavity makes every coverage choice unique, including boundary equalities.
\end{proof}

\subsection{Proof of Proposition \ref{prop:opaque}}
\begin{proof}
Each closed-interval problem in \eqref{eq:branches} has a unique optimiser. Although the actual high branch excludes $\rb$, extending it to that endpoint causes no error in the global maximum: if $r_h=\rb$, its value is strictly below the feasible low value at $\rb$ by \eqref{eq:jump}. Hence it cannot win or tie.

Let $F(b)=U(b)-L(b)$. Uniqueness of each branch optimiser and the envelope theorem give
\[
 F'(b)=r_h(b)-r_l(b)\geq0.
\]
For sufficiently negative $b$, $r_l=0$ and $r_h=\rb$, so $F(b)\to-\infty$ as $b\to-\infty$. For sufficiently positive $b$, $r_l=\rb$ and $r_h=1$, so $F(b)\to+\infty$. A root exists. At any root $r_h>\rb\geq r_l$, and hence $F'>0$. Monotonicity rules out two roots: otherwise $F$ would be identically zero between them. This proves the unique switch and the correspondence.

At $b=G'(0)$, strict convexity gives $br-G(r)<0$ for every $r>0$. The high-branch objective is still lower by $H(1-r)\geq0$, while $L\geq0$ because $r=0$ is feasible. Thus $F<0$ and $b^\dagger>G'(0)$. Since $b_{\min}<0\leq G'(0)$, the switch is attainable by a positive timing advantage.

For $\zeta\leq0$, $W_O(r)\leq W_T(r)$ at every $r$. At $b=G'(1)$, full disclosure has the unique optimum $r=1$, where opacity gives the same value. Every point on the low branch gives a strictly smaller value. Thus $F>0$ and $b^\dagger<G'(1)$. When $\zeta=0$ and $b=G'(\rb)$, full disclosure has unique optimum $\rb$ and the low opaque branch attains it; the high branch is strictly inferior. This gives the sharper lower bound.
\end{proof}

\subsection{Proof of Theorem \ref{prop:main}}
\begin{proof}
Suppose \eqref{eq:strictregion} holds. Proposition \ref{prop:joint} gives a unique $r^*\in(\rb,1)$ and $y^*=\chi(1-r^*)\in(0,r^*)$. At this coverage, the gain over full disclosure is $\zeta y^*>0$. The gain over opacity is also strictly positive: the certified policy induces late adoption after the uncertified observation, with posterior $c$ and social gain $\eta$ per such event, whereas opacity induces no late adoption.

If a restricted policy attained the joint optimum, its coverage would have to equal $r^*$ by uniqueness. Neither restricted report can attain the value there. Since both restricted maxima are attained, the joint maximum is strictly larger than each. Conversely, if $\zeta\leq0$, full disclosure attains the joint optimum. If $\zeta>0$ but $r^*\leq\rb$, opacity attains it; if $r^*=1$, both restricted reports attain it. Proposition \ref{prop:joint} therefore makes \eqref{eq:strictregion} necessary as well as sufficient.

Substitution of the definition of $b$ gives \eqref{eq:tauinterval}. Its width is $2\{G'(1)-G'(\rb)\}>0$ by strict convexity. Its lower bound is positive. This proves nonemptiness directly in primitives, without a selected calibration. Strict inequalities remain valid under sufficiently small changes in the scalar primitives and in $G'$ within the admissible cost class.

At $b=b_C$, joint-optimal coverage is uniquely $\rb$. The low opaque branch attains this value, while every high coverage gives strictly less than the joint optimum. Hence $F(b_C)<0$. At $b=b_F$, joint-optimal coverage is uniquely one and the high opaque branch attains it, whereas the low branch does not. Hence $F(b_F)>0$. Proposition \ref{prop:opaque} proves \eqref{eq:switchbounds}.

For $b\geq b_C$, $b+\zeta>G'(\rb)$, so $r_l=\rb$. In the strict partial region, $G'(r^*)=b-\chi\zeta<b$, giving $r^*<r_T$. The winning high candidate satisfies $r_h\geq r_T$, with strict inequality unless $r_T=1$, because its slope is $b+H>b$. These observations prove the remaining claims, including the two bracketing opaque optima at $b^\dagger$.
\end{proof}

\subsection{Proof of Proposition \ref{prop:order}}
\begin{proof}
The first three claims follow from the clipped inverse \eqref{eq:R}. The high slope is $b+H>b$ and its coverage is bounded below by $\rb$. The low slope is $b+\zeta$, with coverage bounded above by $\rb$. Strictness follows from strict increase of $G'$; the endpoint qualifications account for clipping.

When $\zeta>0$, Theorem \ref{prop:main} gives $b^\dagger>b_C>G'(\rb)$. Therefore the low candidate is the unique opaque optimum throughout $b\leq G'(\rb)$. Above that point, $r_l=\rb$, until the unique switch. This proves the table, including all ties.

For completeness, the following intervals establish both strict directions without numerical examples. All are intervals in $b$, and the affine relation in \eqref{eq:notation} converts them into intervals in $\tau$.

If $\zeta<0$, choose
\[
 G'(0)<b<\min\{G'(0)-\zeta,b^\dagger\}.
\]
The interval is nonempty because $-\zeta>0$ and $b^\dagger>G'(0)$. It gives $r_O=0<r_T$. Choosing $b^\dagger<b<G'(1)$ instead gives the high opaque optimum with $r_O>r_T$; this interval is nonempty by Proposition \ref{prop:opaque}.

If $\zeta=0$, choosing $G'(\rb)<b<b^\dagger$ gives $r_O=\rb<r_T$, whereas $b^\dagger<b<G'(1)$ gives $r_O>r_T$. Both intervals are nonempty by the same proposition.

If $\zeta>0$, choose
\[
 \max\{b_{\min},G'(0)-\zeta\}<b<G'(0).
\]
This interval is nonempty because $b_{\min}<0\leq G'(0)$ and $\zeta>0$. It gives $r_T=0<r_O$. Choosing $G'(\rb)<b<b^\dagger$ instead gives $r_O=\rb<r_T$. This interval is nonempty by \eqref{eq:switchbounds}. Every displayed interval lies above $b_{\min}$, so every construction corresponds to a feasible strictly positive $\tau$.
\end{proof}

\section{Allowing early mixing at the minimum rebate}\label{app:mixing}

At the minimum inducing rebate, an eligible user with a low signal is indifferent. The benchmark selects adoption. We now allow all her best responses and ask whether the welfare and coverage comparisons change. Let $\alpha\in[0,1]$ denote the probability that she adopts at $s_0$. All other early choices are strict. We can represent mixing by an independent Bernoulli draw with probability $\alpha$, used after a low signal. Since the user does not see the report before acting, $\alpha$ cannot depend on its realisation.

\begin{proposition}\label{prop:mixing}
Allowing all early mixed responses at $s_0$ and all privately optimal late responses does not increase the maximum joint-policy value in \eqref{eq:V}. It also does not change the maximum value or optimal coverage correspondence under opacity. Under full disclosure, the maximum over these responses has the unique coverage
\begin{equation}
 \widetilde r_T=\RR\big(\max\{b,v_M\};0,1\big),\qquad
 v_M=\rb(b+\zeta)-\frac{\lambda s_0}{2}(1-\rb).
 \label{eq:mixedT}
\end{equation}
Its value is $W_0+\max_{r\in[0,1]}\{\max(b,v_M)r-G(r)\}$. In the strict-certification region \eqref{eq:strictregion}, $b>v_M$, so $\widetilde r_T=r_T$ and its value equals the fully inducing full-disclosure value. Consequently, the strict welfare superiority in Theorem \ref{prop:main} and its strict implementation remain valid when both restricted alternatives can use early mixed responses.
\end{proposition}

\begin{proof}
At a late posterior equal to $c$, adoption raises social welfare by $\eta>0$. Since no subsequent decision remains, adopting at every such tie weakly maximises welfare across late best responses. It is enough to examine early mixing under this late selection.

Consider a policy with coverage $r$, early mixing $\alpha$, and arbitrary eligibility kernel $\pi$. Set $r'=r\alpha$. For $r'<1$, define a new kernel by
\begin{align}
 \pi'(m\mid1)&=\pi(m\mid1),\notag\\
 \pi'(m\mid0)&=\frac{(1-r)\pi(m\mid0)+r(1-\alpha)\pi(m\mid1)}{1-r\alpha}.
 \label{eq:mixkernel}
\end{align}
The new probabilities are valid and reproduce the joint distribution of $(\theta,a_E,m)$ under full inducement at coverage $r'$. To see why, draw the original eligibility and an independent activation coin before the signal. Call a user effectively eligible when both draws equal one. Among users who are not effectively eligible, \eqref{eq:mixkernel} reproduces the distribution of the original eligibility record. Neither draw depends on the state or signal. The new programme therefore preserves the late user's information and actions, as well as all real adoption payoffs. If $r'=1$, the original policy already induces every eligible user to adopt and no transformation is needed.

The original expected transfer is $s_0r(1+\alpha)/2$, while the new transfer is $s_0r\alpha$. The welfare improvement from the transformation is therefore
\begin{equation}
 G(r)-G(r\alpha)+\frac{\lambda s_0r(1-\alpha)}2\geq0.\label{eq:mixgain}
\end{equation}
The gain is strict for $r>0$ and $\alpha<1$, since administrative costs are strictly convex and nondecreasing. Thus early mixing cannot improve the joint optimum. If the original report was opaque, \eqref{eq:mixkernel} remains opaque, proving the same result for opacity. Rebates below $s_0$ are dominated by no programme; rebates above $s_0$ already induce strict adoption and cannot improve on the minimum inducing cost.

The transformation in \eqref{eq:mixkernel} need not preserve full disclosure. We therefore calculate its value separately. Conditional on an eligible adoption, the posterior is
\[
 \frac{q+\alpha(1-q)}{1+\alpha},
\]
so the late user imitates it exactly when $\alpha\leq\rb$. At an ineligible adoption the posterior remains $q$. For $\alpha\leq\rb$, full-disclosure welfare is
\[
 W_0+\left[\alpha\left\{\eta+\frac\tau2-(1+\lambda/2)s_0\right\}
                 -\frac{\lambda s_0}{2}\right]r-G(r).
\]
The coefficient is linear in $\alpha$. Its best choice is $\alpha=\rb$ if the coefficient on $\alpha$ is positive. Otherwise $\alpha=0$ is best, and no programme weakly dominates. At $\alpha=\rb$, the coverage slope is precisely $v_M$.

For $\alpha>\rb$, eligible observations induce no late adoption, and welfare is
\[
 W_0+\left[\frac\alpha2\{\eta+\tau-(1+\lambda)s_0\}
                    -H-\frac{\lambda s_0}{2}\right]r-G(r).
\]
If the coefficient on $\alpha$ is nonnegative, this branch is maximised at $\alpha=1$, yielding slope $b$. If it is negative, its supremum is approached at $\rb$ from above and is strictly dominated, for $r>0$, by $\alpha=\rb$, where late adoption resumes. Thus only slopes $b$ and $v_M$, together with no programme, can matter. Equation \eqref{eq:mixedT} follows because $r\geq0$ and the administrative cost is common to both choices.

Finally,
\[
 v_M-b=(1-\rb)\left\{\chi\zeta-\frac{\lambda s_0}{2}-b\right\}.
\]
In \eqref{eq:strictregion}, $b>G'(\rb)+\chi\zeta>\chi\zeta-\lambda s_0/2$. Hence $b>v_M$ there. This proves the last claims.
\end{proof}

The global coverage comparison over these weak-equilibrium selections is explicit as well. If $\zeta\leq0$, $\widetilde r_T=r_T$: for $b\leq0$ both candidate slopes give no programme, and for $b>0$ the fully inducing slope is larger. If $\zeta>0$, define
\begin{equation}
 b_X=\min\left\{G'(\rb),\,
 \frac{G'(\rb)+\lambda s_0(1-\rb)/2}{\rb}-\zeta\right\}.
 \label{eq:mixedcross}
\end{equation}
Then the global ordering table in Proposition \ref{prop:order} holds with $r_T$ replaced by $\widetilde r_T$ and its cutoff $G'(\rb)$ replaced by $b_X$. The opaque switch $b^\dagger$ is unchanged. Indeed, $\max(b,v_M)=G'(\rb)$ exactly at $b_X\leq G'(\rb)<b^\dagger$. Below this point, a positive low opaque optimum has marginal slope $b+\zeta$ strictly above either relevant full-disclosure slope; above it and below $b^\dagger$, opaque coverage equals $\rb<\widetilde r_T$. At and above $b^\dagger$ the fully inducing full-disclosure choice is optimal, as shown above. These observations include the zero- and full-coverage equality cases.

The feasible parameter range remains $b\geq b_{\min}$, and $b_X$ can lie below it. For this reason, the claim in Proposition \ref{prop:order} that both strict coverage orderings occur for every fixed $\eta$ applies to fully inducing policies. With early mixing, the adjusted table and \eqref{eq:mixedcross} give the comparison. The main welfare result is unaffected: policies making every reached decision strict can approach the joint optimum and still outperform the best outcome of either restricted report, including outcomes with early mixing.

\section{Private information and implementation}\label{app:exact}

\subsection{Proof of Proposition \ref{prop:general}}
\begin{proof}
We first establish the regularity of the social continuation value. For fixed signal realisation $z$ and cost $t$, its integrand is
\[
 h_{z,t}(p)=w_z(p)\{p_z(p)-t+\eta\}\ind\{p_z(p)\geq t\},
\]
with value zero when $w_z(p)=0$. Where $w_z(p)>0$, the posterior is continuous. At an adoption threshold, selecting adoption includes the positive external benefit, so $h_{z,t}$ is upper semicontinuous. At a zero-weight endpoint, its absolute value tends to zero with $w_z(p)$. The integrands are uniformly bounded. Reverse Fatou's inequality therefore gives upper semicontinuity after integration over costs, and the finite sum over $z$ preserves it. Thus $\LL$ is bounded, nonnegative, and upper semicontinuous on $[0,1]$. If costs have no atoms, the threshold equalities have zero cost probability wherever $w_z(p)>0$. Dominated convergence then gives continuity of $\LL$.

Next consider the posteriors that a report can generate. For message $m$, the adoption masses from the two eligibility groups are
\[
 x_m=\frac{1-r}{2}\pi(m\mid0),\qquad y_m=a_\rho r\pi(m\mid1).
\]
The posterior is $(qx_m+p_\rho y_m)/(x_m+y_m)$, and the masses sum to $(1-r)/2$ and $a_\rho r$. These identities give the stated support and mean. Conversely, given feasible posteriors $p_m$ with conditional probabilities $\nu_m$, set
\[
 x_m=w_\rho(r)\nu_m\frac{p_m-p_\rho}{q-p_\rho},\qquad
 y_m=w_\rho(r)\nu_m\frac{q-p_m}{q-p_\rho}.
\]
Their sums are the required group masses, so dividing by those masses recovers an eligibility-only kernel. The endpoint cases follow by deleting the absent group.

Maximizing the expectation of $\LL$ at a fixed mean gives its concave majorant. For completeness, the set $\{(p,v):p\in[p_\rho,q],\,0\leq v\leq\LL(p)\}$ is compact. At an interior mean, a supporting line to the upper boundary of its convex hull meets the graph at contact points bracketing that mean. Two such points attain the boundary; at an endpoint one point suffices. This proves both attainment and the two-message claim. The majorant is continuous on the closed interval: concavity gives the lower endpoint limits, and upper semicontinuity of the compact hull gives the upper limits.

It remains to choose coverage. Adding early welfare, the contributions following adoption and non-adoption, and fiscal costs gives \eqref{eq:general-program}. The reporting term is concave in $r$: both $w_\rho(r)$ and $w_\rho(r)\mu(\rho r)$ are affine, so concavity of the majorant carries over to its probability-weighted value. This is the perspective property of a concave function. The other benefit terms are affine. Subtracting strictly convex $G$ gives a continuous, strictly concave objective on $[0,1]$, hence an existing and unique optimum.

Finally, changing $\tau$ adds an affine term with coverage slope $\rho\tau/2$. Comparing optimality inequalities at two values of $\tau$ proves weak monotonicity of the maximiser. Compactness, continuity, and uniqueness give its continuity by the usual subsequence argument.
\end{proof}
\subsection{An exact solution with imperfect implementation}\label{sec:compliance}

First take $C=c$ and an uninformative late signal. If $0<\rho\leq\rb$, then $p_\rho\geq c$. Every early adoption induces late adoption under every eligibility report. Reporting is irrelevant and
\begin{equation}
 r^*=\RR(\rho(b+\zeta);0,1),\qquad
 V_\rho=W_0+\max_r\{\rho(b+\zeta)r-G(r)\}.\label{eq:low-compliance}
\end{equation}
Delivery failures can thus leave eligible choices sufficiently informative that the reporting trade-off disappears.

For $\rho>\rb$, define
\begin{align}
 h_\rho&=c-p_\rho>0,& z_\rho&=\eta-h_\rho,\notag\\
 k_\rho&=\frac{q-c}{2h_\rho},&
 b_\rho&=\rho\left(\frac{\eta+\tau-s_0}{2}-\lambda s_0\right)-H,\notag\\
 R_\rho&=\frac{\rb}{\rho}=\frac{k_\rho}{a_\rho+k_\rho}.\label{eq:rho-constants}
\end{align}
An eligible observation now has posterior below the late private threshold. Let $y$ be the unconditional mass of eligible adoptions in the uncertified message. A sufficient certificate has all ineligible users uncertified and
\[
 0\leq y\leq a_\rho r,\qquad y\leq k_\rho(1-r).
\]
Its social value is $W_0+b_\rho r-G(r)+z_\rho y$.

\begin{proposition}\label{prop:rho-exact}
For $\rho>\rb$, the complete joint problem is
\begin{equation}
 \max_{0\leq r\leq1}
 \left\{W_0+b_\rho r-G(r)+(z_\rho)_+
                  \min\{a_\rho r,k_\rho(1-r)\}\right\}.\label{eq:rho-program}
\end{equation}
If $z_\rho\leq0$, full disclosure is optimal and $r^*=\RR(b_\rho;0,1)$. If $z_\rho>0$, set
\[
 A_\rho=b_\rho+a_\rho z_\rho=\rho(b+\zeta),\qquad
 K_\rho=b_\rho-k_\rho z_\rho.
\]
The unique optimal coverage is zero if $A_\rho\leq G'(0)$, is the root of $G'(r)=A_\rho$ below $R_\rho$ when that root lies there, equals $R_\rho$ if $K_\rho\leq G'(R_\rho)\leq A_\rho$, is the root of $G'(r)=K_\rho$ above $R_\rho$ when that root lies below one, and equals one if $K_\rho\geq G'(1)$.

Strictly partial certification beats both optimised full disclosure and optimised opacity exactly when
\begin{equation}
 z_\rho>0,\qquad G'(R_\rho)<K_\rho<G'(1).\label{eq:rho-region}
\end{equation}
This comparison continues to hold when either restricted report can use an early mixed response at $s_0$.
\end{proposition}

\begin{proof}
Merge messages by their late action and move every remaining ineligible observation into the adoption-inducing message. This raises its posterior and adds social surplus $q-c+\eta>0$ per such observation. The other message then certifies eligibility and has posterior $p_\rho<c$. Obedience gives the displayed bound on $y$. Maximizing its linear contribution gives \eqref{eq:rho-program}.

The two derivatives of the objective are $A_\rho-G'(r)$ and $K_\rho-G'(r)$, with a downward change at $R_\rho$. This proves the coverage characterisation. At an interior high-coverage optimum, certification raises welfare strictly above full disclosure because $z_\rho y>0$, and strictly above opacity because opacity induces no late adoption. Uniqueness of joint-optimal coverage and attainment of the restricted maxima give strict global dominance. Outside this region either full disclosure, low-coverage opacity, or full coverage attains the joint maximum. Appendix \ref{app:early-general} proves the stated extension to mixed early choices.
\end{proof}

To show that these conditions hold on a nonempty parameter region, fix $\rho>\rb$ and $z_\rho>0$, and write
\[
 D_\rho=\rho\left(\frac{\eta-s_0}{2}-\lambda s_0\right)-H<0,
 \qquad b_\rho=\rho\tau/2+D_\rho.
\]
Condition \eqref{eq:rho-region} holds precisely when
\begin{equation}
 \frac{2}{\rho}\{G'(R_\rho)-D_\rho+k_\rho z_\rho\}
 <\tau<
 \frac{2}{\rho}\{G'(1)-D_\rho+k_\rho z_\rho\}.\label{eq:rho-tau}
\end{equation}
The lower bound is positive and the interval has width $2\{G'(1)-G'(R_\rho)\}/\rho>0$. For every point in \eqref{eq:strictregion}, continuity also gives a nonempty interval of delivery probabilities below one that preserves strict dominance with the original primitives fixed.

The two restricted rules remain globally solvable. Full disclosure gives $W_{T,\rho}(r)=W_0+b_\rho r-G(r)$. Opaque welfare is
\begin{equation}
 W_{O,\rho}(r)-W_0=
 \begin{cases}
 A_\rho r-G(r),&0\leq r\leq R_\rho,\\
 (b_\rho+H)r-G(r)-H,&R_\rho<r\leq1.
 \end{cases}\label{eq:rho-opaque}
\end{equation}
Optimise each branch with $\RR$, compare their values, and retain both maximisers at equality. Closing the high branch at $R_\rho$ creates no spurious winner: its endpoint value is lower than the true low-branch value by $\eta(1+\rb)/2>0$. Thus the cutoff argument of Section \ref{sec:policy} extends with the nominal credibility threshold $R_\rho=\rb/\rho$. An independent lottery that reveals all eligibility statuses with probability $\xi$ gives $\xi W_{T,\rho}+(1-\xi)W_{O,\rho}$ at fixed coverage; it selects an endpoint unless they tie. It does not implement the eligibility-dependent certificate.

\subsubsection{If delivery itself can be certified}\label{sec:receipt}

If reporting can depend on actual delivery $RD$, let effective coverage be $x=\rho r\in[0,\rho]$. The information problem is then the full-delivery model, with administrative cost $G(x/\rho)$ and rebate financing cost $\lambda s_0x$. This richer instrument can reproduce an eligibility-only report: conditional on $RD=0$, the administrator randomises over the old assignment record with its correct conditional probability. The optimised value can therefore weakly increase, although revealing actual receipt in full need not be optimal.

\subsection{Finite private signals: a complete analytical solution}\label{sec:finite}

Now allow any finite private signal and keep $C=c$. For each realisation with positive probability in at least one state, define its public-posterior adoption threshold
\begin{equation}
 t_z=\frac{c f_0(z)}{(1-c)f_1(z)+c f_0(z)}.\label{eq:signalcut}
\end{equation}
The late user adopts after $z$ exactly when $p\geq t_z$. Thresholds zero or one represent fully revealing realisations and cause no difficulty. Her social continuation value is
\begin{equation}
 \LL(p)=\sum_{z:p\geq t_z}
 \big\{(1-c+\eta)pf_1(z)+(\eta-c)(1-p)f_0(z)\big\}.\label{eq:finiteL}
\end{equation}
It is affine between finitely many thresholds; its jumps, when nonzero, are upward. All private signals are overridden toward adoption if $p\geq\max_z t_z$, and toward non-adoption if $p<\min_z t_z$. Between these bounds the private realisation matters. Comparing signal accuracies alone is not the correct condition when adoption has cost $c$ and the public action is affected by a rebate.

\begin{proposition}\label{prop:finite}
For $\rho>0$, form the ordered finite set
\[
 \mathcal P=\{p_\rho,q\}\cup\{t_z:p_\rho<t_z<q\},
\]
removing duplicates and evaluating $\LL$ with adoption at equality. For any pair $p_i<p_j$ in this set, let $u_{ij}+v_{ij}p$ be the line through $(p_i,\LL(p_i))$ and $(p_j,\LL(p_j))$. Define
\begin{align}
 I_{ij}&=\left[\frac{q-p_j}{\rho(q+p_j-1)},
                   \frac{q-p_i}{\rho(q+p_i-1)}\right],\notag\\
 C_{ij}&=E_0+\frac{u_{ij}+qv_{ij}+\LL(1-q)}2,\notag\\
 k_{ij}&=\rho\left\{E_1-\lambda s_0+
              \frac{u_{ij}+(1-q)v_{ij}-\LL(1-q)}2\right\}.
 \label{eq:finitecandidates}
\end{align}
All these intervals lie in $[0,1]$. For each pair compute
\[
 r_{ij}=\RR(k_{ij};\inf I_{ij},\sup I_{ij}),\qquad
 C_{ij}+k_{ij}r_{ij}-G(r_{ij}).
\]
The largest of these finitely many values is the global joint optimum. The associated pair gives an implementing report. The optimal coverage is unique, although different pairs can implement it. For fully inducing early policies, full disclosure has one affine-benefit coverage problem, and opacity has at most $|Z|+1$ affine-benefit branches; their global optima are likewise finite analytical comparisons.
\end{proposition}

\begin{proof}
Between consecutive points of $\mathcal P$, $\LL$ is affine. At the next threshold its selected value weakly exceeds the left limit. Every interior point of a segment is therefore weakly dominated, at the same mean, by a mixture of its two endpoint values. It suffices to use posteriors in $\mathcal P$. Proposition \ref{prop:general} then gives a best pair.

The mean $\mu(\rho r)$ lies between $p_i$ and $p_j$ exactly on $I_{ij}$, since $\mu$ is decreasing and its inverse is $(q-p)/(q+p-1)$. Multiplying the line $u_{ij}+v_{ij}\mu(\rho r)$ by $w_\rho(r)$ and adding the other terms in \eqref{eq:general-program} gives $C_{ij}+k_{ij}r-G(r)$. Optimizing each feasible pair and comparing all pairs proves the formula.

Under full disclosure the adoption contribution is $(1-r)\LL(q)/2+a_\rho r\LL(p_\rho)$, and the non-adoption contribution is $(1-\rho r)\LL(1-q)/2$. Both are affine. Under opacity use $\LL(\mu(\rho r))$ instead of its concave majorant. On each signal-action region, multiplying by $w_\rho(r)$ again produces an affine function. Optimise the finitely many regions and compare values. Lower endpoint values created by closing an excluded branch cannot exceed the actual threshold value available in an adjacent branch. Thus this closure does not create a false maximum.
\end{proof}

This procedure compares all relevant candidates and gives the global solution. With private information, the optimal report can induce different responses to the private signal after its two messages; it need not retain the benchmark certificate.

\subsubsection{A region in which the private signal affects behaviour}\label{sec:active-signal}

We next construct a parameter region in which the late user acts on her private signal after an optimal report. Set $\rho=1$ and let a binary late signal have accuracy $\omega\in(1/2,c)$. Its two public thresholds are
\begin{equation}
 t_h(\omega)=\frac{c(1-\omega)}{(1-c)\omega+c(1-\omega)},\qquad
 t_\ell(\omega)=\frac{c\omega}{(1-c)(1-\omega)+c\omega}.
 \label{eq:binarycuts}
\end{equation}
They satisfy $1/2<t_h<c<t_\ell$. Fix $c$ and choose
\[
 t_\ell(c)=\frac{c^2}{(1-c)^2+c^2}<q<1,
\]
which is possible. For $\omega<c$ sufficiently close to $c$, both thresholds lie below $q$. At $t_h$, only the favourable private signal induces adoption and
\[
 \LL(t_h)=\eta\{\omega t_h+(1-\omega)(1-t_h)\}>0.
\]
As $\omega\uparrow c$, $t_h\downarrow1/2$, so
\begin{equation}
 m_h:=\frac{\LL(t_h)}{t_h-1/2}\longrightarrow+\infty.\label{eq:signal-slope}
\end{equation}
The corresponding ratios at $t_\ell$ and $q$ stay finite. Hence, on a nonempty interval of accuracies, the first segment of the concave majorant joins $(1/2,0)$ to $(t_h,\LL(t_h))$ strictly above the full-disclosure chord.

Put $R_h=(q-t_h)/(q+t_h-1)$ and
\[
 D_h=\frac{\eta-s_0}{2}-\lambda s_0-\frac{q-1/2}{2}m_h.
\]
On $r\in[R_h,1]$, the joint objective has coverage slope $\tau/2+D_h$ before administrative costs. Choose
\begin{equation}
 2\{G'(R_h)-D_h\}<\tau<2\{G'(1)-D_h\}.\label{eq:active-signal-family}
\end{equation}
Its width is positive. Its lower bound is positive when $\omega$ is sufficiently close to $c$, because $D_h\to-\infty$. Strict concavity then gives a unique optimum in $(R_h,1)$. Optimal reporting mixes public posteriors $1/2$ and $t_h$: after the latter, the late user adopts on her favourable private signal and declines on the unfavourable one. Opacity induces no adoption at the optimal public mean, while full disclosure lies strictly below the concave majorant there. Both optimised restricted rules are strictly inferior by uniqueness and attainment. Appendix \ref{app:early-general} covers their possible early mixed responses.

A slightly higher uncertified posterior, still below $t_\ell$, makes both private-signal responses strict. The implementing probabilities and the right-hand continuation value converge to those just constructed, so a sufficiently small increase preserves the strict welfare gain. A slightly higher rebate makes early inducement strict as well.

This construction establishes that selective reporting can strictly improve welfare even when private information affects the late user's action. It gives a nonempty parameter family, rather than a ranking for every signal precision. In particular, with a perfectly revealing private signal, public reports affect neither her decision nor her social contribution, and the strict reporting gain disappears.

\subsection{An exact solution with dispersed late costs}\label{sec:uniform}

Take an uninformative late signal and $C$ uniform on $[c-\varepsilon,c+\varepsilon]$, with $\varepsilon>0$. Allow delivery failures simultaneously. Assume
\begin{equation}
 p_\rho<c-\varepsilon<c+\varepsilon<q.\label{eq:uniformsupport}
\end{equation}
Thus a certified eligible observation induces no adoption for any late cost, while an ineligible adoption induces adoption for all costs. The distribution has mean $c$ and no atom at any individual indifference point.

Writing $t_-=c-\varepsilon$ and $t_+=c+\varepsilon$, the social continuation value is
\begin{equation}
 \LL_\varepsilon(p)=
 \begin{cases}
 0,&p\leq t_-,\\
 \dfrac{(p-t_-)^2+2\eta(p-t_-)}{4\varepsilon},&t_-<p<t_+,\\[4pt]
 p-c+\eta,&p\geq t_+.
 \end{cases}\label{eq:uniformL}
\end{equation}
It is continuous, strictly convex on the middle interval, and linear above it. Its derivative falls at $t_+$.

\begin{proposition}\label{prop:uniform}
Under \eqref{eq:uniformsupport}, retain $b_\rho$ and $z_\rho=\eta+p_\rho-c$ from \eqref{eq:rho-constants}, and define
\begin{equation}
 k_{\rho,\varepsilon}=\frac{q-c-\varepsilon}{2(c+\varepsilon-p_\rho)},\qquad
 R_{\rho,\varepsilon}
   =\frac{q-c-\varepsilon}{\rho(q+c+\varepsilon-1)}.
 \label{eq:uniformconstants}
\end{equation}
If $z_\rho\leq0$, full disclosure is optimal. If $z_\rho>0$, a sufficient optimal certificate leaves
\begin{equation}
 y=\min\{a_\rho r,k_{\rho,\varepsilon}(1-r)\}\label{eq:uniformy}
\end{equation}
eligible adoption mass uncertified. Whenever this is strictly between zero and $a_\rho r$, its public posterior is $c+\varepsilon$, so every late cost type adopts except for an immaterial endpoint tie. The exact joint coverage problem is
\begin{equation}
 \max_r\left\{W_0+b_\rho r-G(r)+(z_\rho)_+
       \min\{a_\rho r,k_{\rho,\varepsilon}(1-r)\}\right\}.\label{eq:uniformprogram}
\end{equation}
Strict partial certification beats both optimised restricted reports precisely when
\begin{equation}
 z_\rho>0,\qquad
 G'(R_{\rho,\varepsilon})
       <b_\rho-k_{\rho,\varepsilon}z_\rho<G'(1).\label{eq:uniformregion}
\end{equation}
It also beats their best outcomes with early mixing. The five coverage cases in Proposition \ref{prop:rho-exact} apply with $k_\rho,R_\rho$ replaced by $k_{\rho,\varepsilon},R_{\rho,\varepsilon}$.
\end{proposition}

\begin{proof}
Integrating $(p-C+\eta)\mathbf1\{C\leq p\}$ gives \eqref{eq:uniformL}. Its concave majorant on $[p_\rho,q]$ can be computed directly. On the middle interval, the derivative of
\[
 \LL_\varepsilon'(p)(p-p_\rho)-\LL_\varepsilon(p)
\]
is $\LL_\varepsilon''(p)(p-p_\rho)>0$. Therefore the ratio $\LL_\varepsilon(p)/(p-p_\rho)$ has no interior maximum there; its maximum is at an endpoint. The ratio is zero at $t_-$, while above $t_+$ it is $1+z_\rho/(p-p_\rho)$.

If $z_\rho>0$, the largest ratio is at $t_+$. The majorant is the line from $(p_\rho,0)$ to $(t_+,\varepsilon+\eta)$ followed by $p-c+\eta$ up to $q$. The first slope is greater than one, so the proposed majorant is concave and minimal. If $z_\rho\leq0$, the largest ratio is at $q$, and the endpoint chord implements full disclosure.

In the first case, use opacity when the public mean is at least $t_+$; otherwise mix $p_\rho$ and $t_+$. Their implementation gives \eqref{eq:uniformy}, and their joining coverage is \eqref{eq:uniformconstants}. All types adopt at the uncertified observation, so expected cost there remains $c$. The welfare coefficients $b_\rho,z_\rho$ are consequently unchanged, giving \eqref{eq:uniformprogram}.

The coverage derivatives and global comparison now follow as in Proposition \ref{prop:rho-exact}. For a mean strictly between $p_\rho$ and $t_+$, the majorant strictly exceeds \eqref{eq:uniformL}, so opacity is strictly inferior at the unique partial-certification optimum. Full disclosure is also strictly inferior there. Conversely, elsewhere an endpoint reporting rule attains the joint optimum. The mixed-response conclusion follows from Appendix \ref{app:early-general}.
\end{proof}

The conditions hold on a nonempty region. For admissible $\rho,\varepsilon$ satisfying \eqref{eq:uniformsupport} and $z_\rho>0$, replace $k_\rho,R_\rho$ in \eqref{eq:rho-tau} by $k_{\rho,\varepsilon},R_{\rho,\varepsilon}$. The resulting interval for the timing benefit is positive and has nonzero width. Hence the exact comparison allows positive cost dispersion and a positive delivery-failure probability.

The opaque welfare function is now continuous in coverage. Its global maximiser need not change continuously with policy benefits. Section \ref{sec:jump} proves that the policy jump from Section \ref{sec:policy} can survive this smoothing.

\subsection{Early indifference and the comparison set}\label{app:early-general}

The following observation allows both restricted reporting rules to use all early best responses. It applies to every late information and cost specification above.

\begin{lemma}\label{lem:mixing}
Let $\alpha\in[0,1]$ be the probability that a low-signal user with a usable rebate adopts at $s_0$. Any policy with $r>0$ and $\alpha<1$ is strictly dominated, in the joint reporting problem, by a fully inducing policy at coverage $r\alpha$. Under opacity the dominating policy is also opaque. Consequently, if the unique joint optimum strictly beats the two fully inducing restricted-report optima, it also strictly beats their best outcomes over all early mixtures and all late best responses.
\end{lemma}

\begin{proof}
Represent early mixing by a Bernoulli activation coin independent of the state, signals, delivery, and assignment, used only after a low signal. Effective assignment is the conjunction of original assignment and this coin. It has probability $r\alpha$, and its conjunction with delivery produces exactly the original early action. Define the new eligibility-only kernel by
\begin{align}
 \pi'(m\mid1)&=\pi(m\mid1),\notag\\
 \pi'(m\mid0)&=
 \frac{(1-r)\pi(m\mid0)+r(1-\alpha)\pi(m\mid1)}{1-r\alpha}.
 \label{eq:mixkernel-general}
\end{align}
It reproduces the joint distribution of $(\theta,a_E,m)$, hence also of late private information, decisions, and real social payoffs. The original expected transfer is $s_0\rho r(1+\alpha)/2$, and the new one is $s_0\rho r\alpha$. The welfare gain is
\[
 G(r)-G(r\alpha)+\frac{\lambda s_0\rho r(1-\alpha)}2>0.
\]
Strictness follows from strictly convex, nondecreasing $G$. If the old report is constant, the new one is constant.

For the last statement, subinducing and superinducing payments reduce to no programme or the minimum-inducing comparison. At a late tie, adoption weakly maximises social welfare because it contributes $\eta>0$. For either restricted report, the probabilities of each state and public history are continuous in $(r,\alpha)$. At a positive-probability history, the posterior is continuous, and its probability-weighted continuation value is upper semicontinuous by the regularity argument in the proof of Proposition \ref{prop:general}. If a history disappears, its contribution tends to zero because $\LL$ is bounded. Adding the continuous early and fiscal terms gives an upper-semicontinuous objective on the compact square $(r,\alpha)\in[0,1]^2$. Each restricted problem therefore attains its maximum, including when costs have atoms. A fully inducing restricted policy is strictly below the assumed joint value. A positive-coverage mixed policy is also strictly below it by the transformation. No programme is below it because it is feasible under both restricted reports. Attainment turns these pointwise strict inequalities into a strict comparison with each restricted maximum.
\end{proof}

Early mixing can still improve full disclosure relative to full inducement: the transformation in \eqref{eq:mixkernel-general} may turn a fully revealing report into a partial one. The lemma shows that allowing this improvement cannot close a strict gap to the joint optimum. Strict implementation is established in Section \ref{sec:strict} for the benchmark, in Appendix \ref{sec:active-signal} for the active-private-signal construction, and in Proposition \ref{prop:robust} for simultaneous small perturbations. These arguments use posterior margins that need not be available for arbitrary late cost distributions or signal structures.

\section{Simultaneous small perturbations}\label{sec:local}

We now allow delivery failures, private late information, and cost dispersion to vary together. The argument uses bounds on the social continuation value and allows arbitrary shapes of the cost distribution within a small support.

Suppose $C$ is supported on $[c-\varepsilon,c+\varepsilon]$, without a mean restriction. Suppose every private likelihood ratio satisfies
\begin{equation}
 e^{-\nu}\leq\frac{f_1(z)}{f_0(z)}\leq e^\nu,\qquad
 \delta=\varepsilon+\frac\nu4.\label{eq:weak-signal}
\end{equation}
For this local statement both state probabilities are positive for each retained realisation. The signal may still be informative. Since the derivative of the logistic function is at most $1/4$, \eqref{eq:privateposterior} implies
\begin{equation}
 |p_z(p)-p|\leq\frac\nu4\quad\text{at every public posterior}.\label{eq:posterior-bound}
\end{equation}
Let $\LL_t(p)=(p-t+\eta)\mathbf1\{p\geq t\}$ denote social continuation value with no private signal and deterministic late cost $t$; early cost remains $c$.

\begin{lemma}\label{lem:sandwich}
For $\delta>0$,
\begin{equation}
 \LL_{c+2\delta}(p)\leq\LL(p)\leq\LL_{c-\delta}(p)
 \qquad\text{for every }p\in[0,1].\label{eq:pointwise-sandwich}
\end{equation}
\end{lemma}

\begin{proof}
If the late user adopts, \eqref{eq:posterior-bound} and $C\geq c-\varepsilon$ imply $p\geq c-\delta$. Her conditional social gain is at most $p-(c-\delta)+\eta$. Since this bound is nonnegative, summing it only over adopters gives the upper inequality.

If $p\geq c+2\delta$, then $p_z(p)\geq c+2\varepsilon+\nu/4>c+\varepsilon$ whenever $\delta>0$. All actual late decisions are strictly to adopt. Their expected social contribution is $p-\E C+\eta\geq p-(c+2\delta)+\eta$, using Bayesian consistency $\E[p_Z(p)\mid p]=p$ and independence of cost. At a posterior where the deterministic high-cost user declines, the comparison value is zero; every actual adoption has nonnegative private surplus plus $\eta$. This proves the lower inequality.
\end{proof}

For the no-private-signal problem with delivery $\rho$ and deterministic late cost $t$, let $\widehat V_\rho(t)$, $\widehat T_\rho(t)$, and $\widehat O_\rho(t)$ denote optimised joint, full-disclosure, and opaque values. The two restricted values include all early mixed responses. Appendix \ref{app:envelopes} gives explicit scalar formulas for all three. They are continuous near $(\rho,t)=(1,c)$ whenever $p_\rho<t<q$.

\begin{proposition}\label{prop:robust}
Fix any benchmark point in \eqref{eq:strictregion}. For all sufficiently small simultaneous delivery failure $1-\rho$, cost support width $\varepsilon$, and private likelihood-ratio bound $\nu$, optimal reporting strictly beats both restricted reports, even when those reports can use early mixing. For $\delta>0$, a sufficient analytical condition is
\begin{align}
 p_\rho&<c-\delta<c+2\delta<q,\notag\\
 \widehat V_\rho(c+2\delta)
 &>\max\{\widehat T_\rho(c-\delta),\widehat O_\rho(c-\delta)\}
       +\lambda\epsilon_s,\label{eq:robust-condition}
\end{align}
where $\epsilon_s>0$ is an arbitrarily small addition to the early rebate. Such parameters form a nonempty neighborhood on the feasible side of full delivery. A certificate implementing the high deterministic-cost optimum, with rebate $s_0+\epsilon_s$, then beats both restricted maxima while every reached individual decision is strict.

As $(\rho,\varepsilon,\nu)\to(1,0,0)$, the three optimised values converge to their benchmark values. Joint-optimal coverage converges to $r^*$. Full-disclosure optimal coverage, allowing early mixing, converges to $r_T$. Every limit point of optimal opaque coverage belongs to the benchmark opaque optimum set. Away from its switch, this means convergence to its unique coverage choice.
\end{proposition}

\begin{proof}
The early decision problem and feasible public kernels do not depend on the late information perturbation. Applying \eqref{eq:pointwise-sandwich} to each public history, and then optimising over the same admissible policies and early responses, gives, for $X\in\{V,T,O\}$,
\begin{equation}
 \widehat X_\rho(c+2\delta)\leq X_{\rho,F_C,Z}
                              \leq\widehat X_\rho(c-\delta).\label{eq:value-sandwich}
\end{equation}
The deterministic formulas in Appendix \ref{app:envelopes} are continuous at the benchmark. The bounds therefore establish convergence of the optimised values. At $(1,c)$ the joint gap is $\Gamma>0$. Continuity gives a nonempty set with \eqref{eq:robust-condition}, including positive $\varepsilon,\nu,1-\rho$ and sufficiently small $\epsilon_s$.

Near the benchmark, the high deterministic-cost joint optimum is still in its strict-certification region. Its uncertified posterior is $c+2\delta$, while its certified posterior $p_\rho$ is below $c-\delta$. Non-adoption has posterior $1-q<p_\rho$. Equation \eqref{eq:posterior-bound} consequently makes all reached late actions strict. Paying $s_0+\epsilon_s$ makes the early induced actions strict and costs at most $\lambda\epsilon_s$. The resulting welfare is at least the left side of \eqref{eq:robust-condition} minus that cost, while the right side bounds the best restricted alternatives.

To establish convergence of coverage, first restrict the joint problem to full inducement; Lemma \ref{lem:mixing} permits this without loss when optimal coverage is positive. The pointwise bounds hold before maximising over $r$. The deterministic optimised-reporting objectives are jointly continuous in $(\rho,t,r)$ near the benchmark on a compact set. They therefore converge uniformly in $r$. The actual optimised-reporting objective does too. Its maximisers converge to the unique benchmark $r^*$.

For opacity, consider a convergent sequence $r_n\to r$. Away from $r=\rb$, private posteriors and costs eventually put every late user on the benchmark side of the adoption threshold. At $r=\rb$, late private surplus tends to zero and the limiting adoption probability is at most one, so the limsup of welfare is at most the selected benchmark value, which includes the entire external benefit. This proves the required upper bound along every convergent policy sequence. Together with optimised-value convergence, it forces every limiting maximiser into the benchmark opaque optimum set.

The same upper-bound argument applies to full disclosure on the compact set $(r,\alpha)$. In \eqref{eq:strictregion}, its benchmark optimum is uniquely fully inducing at $r_T>0$: the alternative mixed-response slope in Appendix \ref{app:mixing} is strictly smaller. Thus full-disclosure coverage converges to $r_T$, and any optimal early mixing probabilities converge to one. This completes the proof.
\end{proof}

The bound on likelihood ratios is important for this local result. It permits informative private signals, but excludes rare observations with arbitrarily large likelihood ratios. For stronger signals, Section \ref{sec:finite} still provides an exact solution, although the welfare ranking need not hold. The cost distribution may be asymmetric and need not have mean $c$.

Strict coverage comparisons also persist near the benchmark. Below the opaque switch, within \eqref{eq:strictregion}, sufficiently small perturbations preserve $r_O<r^*<r_T$. Above the switch, both restricted optima remain above optimal joint coverage. Their ordering relative to each other persists when it was strict at the benchmark. At the switch itself, continuity does not select one of the two opaque optima.

\subsection{A policy jump with continuous welfare}\label{sec:jump}

With a continuous distribution of late costs, opaque welfare is continuous in coverage. The jump in \eqref{eq:rho-opaque} disappears because different cost types stop imitating at different posteriors. Nevertheless, maximising this continuous function can still lead the government to switch between separated programme sizes.

Fix $(q,c,\eta,\lambda,G)$ with $\eta>d$ and vary the early timing benefit through $b$. The full-delivery benchmark has a unique opaque switch $b^\dagger$ strictly between
\[
 b_C=G'(\rb)+\chi\zeta,\qquad b_F=G'(1)+\chi\zeta.
\]
At that switch its two opaque optima are $\rb$ and $r_h^\dagger>\rb$, while joint-optimal coverage satisfies
\[
 \rb<r^*(b^\dagger)<r_h^\dagger.
\]
These coverage comparisons follow from the global optimum characterised in Section \ref{sec:policy}.

\begin{proposition}\label{prop:jump}
There are a neighborhood of $b^\dagger$ and disjoint closed coverage intervals $I_l,I_h$ containing $\rb,r_h^\dagger$, respectively, such that the following holds for all sufficiently small perturbations covered by Proposition \ref{prop:robust}. Assume in addition that $F_C$ has no atoms, so opaque welfare is continuous in coverage.

Within that neighborhood there is a unique switch $b^\dagger_{\rho,F_C,Z}$. Below it all globally optimal opaque coverages lie in $I_l$; above it all lie in $I_h$. At the switch, there is at least one optimum in each interval. The intervals have a strictly positive gap. The switch converges to $b^\dagger$ as the perturbations vanish.

Joint-optimal coverage remains unique and continuous in $b$. Near the switch it lies strictly between the two intervals, and joint reporting has strictly higher welfare than either optimised restricted report. No uniqueness of the opaque maximiser within either interval is asserted.
\end{proposition}

\begin{proof}
Choose disjoint small closed intervals around $\rb$, $r^*(b^\dagger)$, and $r_h^\dagger$, in that order, with each point in the corresponding relative interior in $[0,1]$. At the benchmark switch the first and third contain the only opaque maximisers. Upper semicontinuity gives a strictly positive welfare gap on the compact complement of their relative interiors. By the explicit branch solutions of Section \ref{sec:policy} and continuity of the joint optimum, choose $\kappa>0$ small enough that, for $b\in[b^\dagger-\kappa,b^\dagger+\kappa]\subset(b_C,b_F)$, all benchmark opaque maximisers remain inside the first or third interval and the joint maximiser remains inside the middle interval. At the left endpoint the opaque optimum is in the low interval only; at the right it is in the high interval only.

For sufficiently small perturbations, the optima remain in these intervals. We prove this by compactness, since opaque welfare need not converge uniformly through the benchmark discontinuity. If they failed, take a sequence of perturbations tending to zero, corresponding benefit parameters tending to a point in the stated compact neighborhood, and violating maximising coverages tending to a limit. The sequence-wise upper bound in Proposition \ref{prop:robust}, together with value convergence, makes that limit a benchmark maximiser. This contradicts the positive gap to the complement of the chosen intervals. Value convergence is uniform over this benefit neighborhood: after removing the coverage-independent term, all objectives have coverage slope $\rho r$ in $b$, bounded by one, so the optimised values are equicontinuous in $b$. The same compactness argument and uniqueness apply to joint coverage.

We also need convergence of the maxima restricted to each coverage interval at the two endpoints of the benefit neighbourhood. In the low interval, the benchmark maximiser is $\rb$. Consider the fully inducing opaque policy with coverage
\[
 r_{\mathrm{rec}}(\rho,\delta)
 =\frac{q-c-2\delta}{\rho(q+c+2\delta-1)}.
\]
For sufficiently small positive $\delta$ and $1-\rho$, this coverage belongs to $I_l$, converges to $\rb$, and gives adoption posterior $\mu(\rho r_{\mathrm{rec}})=c+2\delta$. By \eqref{eq:posterior-bound}, every late cost and signal type then strictly adopts after observing adoption. After non-adoption, every type strictly declines for sufficiently small perturbations. This policy's welfare converges to the selected low benchmark value. In $I_h$, keep the benchmark high maximiser at the given benefit endpoint fixed. Its adoption posterior is strictly below $c$, so all late types eventually decline and its welfare converges to the high benchmark value. These two recovery policies provide the lower bounds for the interval maxima. The sequence-wise upper bound used above gives their upper bounds. Hence both interval values converge, preserving the opposite signs of the high-minus-low value difference at the two benefit endpoints.

Fix one sufficiently small perturbation. After dropping a common term independent of coverage, write its opaque objective as
\[
 f(r)+\rho br.
\]
The distribution of private information, the cost distribution, and delivery are held fixed while $b$ varies. Let $L(b)$ and $U(b)$ be its maxima on $I_l$ and $I_h$. They are continuous, because the objective is affine in $b$ with uniformly bounded slopes. For $b_2>b_1$,
\begin{align*}
 U(b_2)-U(b_1)&\geq\rho(\inf I_h)(b_2-b_1),\\
 L(b_2)-L(b_1)&\leq\rho(\sup I_l)(b_2-b_1).
\end{align*}
Consequently $U-L$ is strictly increasing. It is negative at the left endpoint and positive at the right endpoint, so it has one and only one zero. Localization ensures these are the only globally relevant alternatives. At the zero both compact-interval maxima are attained, giving separated globally optimal coverages. A sequence of such switches can converge only to the benchmark switch; otherwise localisation at a benefit with a unique benchmark branch gives a contradiction.

Proposition \ref{prop:general} gives a unique joint optimum continuous in $b$. Its localisation in the middle interval proves the coverage claim. Finally, the benchmark strict welfare gap is continuous and positive on a sufficiently small closed benefit neighborhood inside $(b_C,b_F)$. The bounds on optimised values preserve a strictly positive gap uniformly over this neighbourhood. Lemma \ref{lem:mixing} includes the enlarged restricted comparison set.
\end{proof}

Thus continuity of welfare in coverage does not guarantee continuity of the government's choice. Opacity leaves two separated programmes competing for the global maximum. Selective reporting permits intermediate coverage and delivers higher welfare near the switch.

The proposition concerns small perturbations and timing benefits near the benchmark switch. With arbitrary cost distributions, it need not yield a single global switch, a differentiable opaque optimum, or the benchmark's exact cutoff formulas.

\section{Value bounds for the robustness proof}\label{app:envelopes}

This appendix gives the analytical formulas for the value bounds used in Proposition \ref{prop:robust}. Keep early cost $c$ and rebate $s_0$ fixed. Vary only deterministic late cost $t$, with $p_\rho<t<q$. Define
\begin{align}
 H_t&=\frac{q-t+\eta}{2},& W_0(t)&=E_0+H_t,\notag\\
 B_E&=E_1-\lambda s_0,& A_t&=B_E+\frac{1-q-t+\eta}{2},\notag\\
 b_\rho(t)&=\rho B_E-H_t,& z_\rho(t)&=\eta+p_\rho-t,\notag\\
 k_\rho(t)&=\frac{q-t}{2(t-p_\rho)},&
 R_\rho(t)&=\frac{q-t}{\rho(q+t-1)}\in(0,1).
 \label{eq:envelopeconstants}
\end{align}
The fully optimised joint value is
\begin{equation}
 \widehat V_\rho(t)=W_0(t)+\max_{r\in[0,1]}
 \left\{b_\rho(t)r-G(r)+(z_\rho(t))_+
                    \min\{a_\rho r,k_\rho(t)(1-r)\}\right\}.
 \label{eq:hatV}
\end{equation}
Its optimiser is given by the two-slope characterisation of Proposition \ref{prop:rho-exact}, with the constants in \eqref{eq:envelopeconstants}. For opacity set
\begin{align*}
 r_l(t)&=\RR(\rho A_t;0,R_\rho(t)),\\
 r_h(t)&=\RR(\rho B_E;R_\rho(t),1).
\end{align*}
Then
\begin{equation}
 \widehat O_\rho(t)=W_0(t)+
 \max\{\rho A_t r_l-G(r_l),\ \rho B_Er_h-G(r_h)-H_t\}.
 \label{eq:hatO}
\end{equation}
As before, a closed high-branch candidate at $R_\rho(t)$ is strictly below the true endpoint value and cannot spuriously win.

To include early mixing under full disclosure, let
\begin{equation}
 v_{M,\rho}(t)=\frac{q-t}{q+t-1}
                   \left(A_t+\frac{\lambda s_0}{2}\right)
                         -\frac{\lambda s_0\rho}{2}.
 \label{eq:mixedslope}
\end{equation}
The exact value is
\begin{equation}
 \widehat T_\rho(t)=W_0(t)+\max_{r\in[0,1]}
           \{\max[b_\rho(t),v_{M,\rho}(t)]r-G(r)\}.
 \label{eq:hatT}
\end{equation}
Its coverage is $\RR(\max[b_\rho(t),v_{M,\rho}(t)];0,1)$.

For clarity, the mixed-response formula follows directly from the two late-action regions. If a usable low-signal rebate induces adoption with probability $\alpha$, the posterior following an eligible adoption is $\mu(\rho\alpha)$. The late user imitates this observation exactly when
\[
 \alpha\leq\frac{q-t}{\rho(q+t-1)}=R_\rho(t).
\]
On this region the full-disclosure coverage slope is
\[
 \rho\left\{\alpha\left(A_t+\frac{\lambda s_0}{2}\right)
                                    -\frac{\lambda s_0}{2}\right\}.
\]
It is linear in $\alpha$, so its best useful endpoint is $\alpha=R_\rho(t)$, giving \eqref{eq:mixedslope}; otherwise no programme dominates. Above the threshold, eligible observations induce no late adoption. The coverage slope is again linear in $\alpha$. It is maximised either at full inducement, giving $b_\rho(t)$, or at the excluded threshold from above, whose value is dominated by the included imitation endpoint. This proves \eqref{eq:hatT}.

All constants, endpoints, and clipped inverse solutions in \eqref{eq:hatV}--\eqref{eq:hatT} are continuous on a neighborhood of $(1,c)$ with $p_\rho<t<q$. The maximum of the finitely many continuous branch values is continuous. Equivalently, \eqref{eq:hatV} has a jointly continuous objective on a fixed compact coverage domain. These observations justify the continuity used in the robustness proof without presuming continuity of opaque welfare at a fixed threshold coverage.

At $\rho=1,t=c$, the joint and opaque values reduce to Section \ref{sec:policy}, while the full-disclosure value, which allows early mixing, reduces to Appendix \ref{app:mixing}. In particular, inside \eqref{eq:strictregion}, the fully inducing full-disclosure slope strictly exceeds its mixed alternative. This is why the limiting full-disclosure coverage in Proposition \ref{prop:robust} is unambiguous.

\end{document}